%% file: CameraReady.tex
\documentclass[sigconf,screen]{acmart}

\newcommand{\asplossubmissionnumber}{561}

\usepackage{mathtools,amsthm}
\usepackage{tikz}
\usepackage{flushend}

\usepackage{listings}
\usepackage{makecell}
\usepackage{wasysym}
\allowdisplaybreaks
\newsavebox{\tablebox}
\usetikzlibrary{quantikz}

\input{pmymacro}

\newtheorem{theorem}{Theorem}[section]
\newtheorem{lemma}{Lemma}[section]
\newtheorem{definition}{Definition}[section]
\newtheorem{example}{Example}[section]
\setcopyright{acmlicensed}
\acmPrice{15.00}
\acmDOI{10.1145/3582016.3582039}
\acmYear{2023}
\copyrightyear{2023}
\acmSubmissionID{asplosc23main-p561-p}
\acmISBN{978-1-4503-9918-0/23/03}
\acmConference[ASPLOS '23]{Proceedings of the 28th ACM International Conference on Architectural Support for Programming Languages and Operating Systems, Volume 3}{March 25--29, 2023}{Vancouver, BC, Canada}
\acmBooktitle{Proceedings of the 28th ACM International Conference on Architectural Support for Programming Languages and Operating Systems, Volume 3 (ASPLOS '23), March 25--29, 2023, Vancouver, BC, Canada}
\received{2022-10-20}
\received[accepted]{2023-01-19}

\title{
Verification of Nondeterministic Quantum Programs}

\author{Yuan Feng}
\orcid{0000-0002-3097-3896}
\email{Yuan.Feng@uts.edu.au}
\affiliation{
    \institution{University of Technology Sydney}
    \country{Australia}
}

\author{Yingte Xu}
\orcid{0000-0001-9071-7862}
\email{Yingte.Xu@student.uts.edu.au}
\affiliation{
    \institution{University of Technology Sydney}
    \country{Australia}
}



\begin{abstract}
Nondeterministic choice is a useful program construct that provides a way to describe the behaviour of a program without specifying the details of possible implementations. It supports the stepwise refinement of programs, a method that has proven useful in software development. Nondeterminism has also been introduced in quantum programming, and termination of nondeterministic quantum programs has been extensively analysed. In this paper, we go beyond  termination analysis to investigate the verification of nondeterministic quantum programs where properties are given by sets of hermitian operators on the associated Hilbert space. Hoare-type logic systems for partial and total correctness are proposed which turn out to be both sound and relatively complete with respect to their corresponding semantic correctness. To show the utility of these proof systems, we analyse some quantum algorithms such as quantum error correction scheme, Deutsch algorithm, and a nondeterministic quantum walk. Finally, a proof assistant prototype is implemented to aid in the automated reasoning of nondeterministic quantum programs. 
\end{abstract}

\begin{document}

%\date{}
\maketitle

\thispagestyle{empty}

	\section{Introduction}
	
	The introduction of nondeterminism into a programming language allows one to describe both a specification and its possible implementations within the same language. It naturally supports the technique of stepwise refinement of programs, which has proven useful in the development of computer software~\cite{morgan1993refinement,back2012refinement,morris1987theoretical}. From the verification perspective, it provides a way to reason about the correctness of an abstract specification before it is fully implemented~\cite{dijkstra1976discipline}, making it possible to detect errors  in the earliest stages of software development.
			
	Programs with both nondeterministic and probabilistic choices have been investigated in~\cite{morgan1996probabilistic,jifeng1997probabilistic,McIver:2001,mciver2005abstraction}. Interestingly, as pointed out in~\cite{jifeng1997probabilistic}, there are two natural models to describe the way nondeterminism interacts with probabilistic choice.	This difference can be best illustrated by regarding the execution of a nondeterministic and probabilistic program as a game against a malicious adversary. In the relational model, the adversary makes his nondeterministic choice during runtime execution, taking advantage of all the probabilistic choices already made up to that point. By contrast, the adversary in the lifted model must make all the nondeterministic choices in compile-time, before any real execution of the program. It was shown in~\cite{jifeng1997probabilistic} that programs in the relational model enjoy many nice algebraic properties which are useful in static analysis and compiler design~\cite{hoare1993normal}. 
	%This makes the relational model more attractive than the lifted one.

	Nondeterminism was also introduced into quantum programming many years ago.  Zuliani~\cite{zuliani2004non} extended the quantum Guarded-Command Language qGCL~\cite{sanders2000quantum} with a nondeterministic choice construct, for the purpose of modelling and reasoning about counterfactual computations~\cite{mitchison2001counterfactual} and quantum mixed-state systems. Termination analysis of nondeterministic quantum programs was first studied in~\cite{li2014termination}, where programs are described in a language-independent way: a nondeterministic quantum program simply consists of a set of deterministic quantum Markov chains (mathematically expressed as super-operators) with the same Hilbert space and a two-outcome measurement to determine whether the program has terminated. Later on, termination analysis was extended to a quantum programming language which involves both demonic and angelic nondeterminism, using the technique of  linear ranking super-martingale~\cite{li2017algorithmic}. 

	Despite the termination analysis in~\cite{li2014termination,li2017algorithmic}, the verification of nondeterministic quantum programs is rarely addressed in the literature. In this paper, we consider an extended while-language for quantum programs where a construct for binary (demonic) nondeterministic choice is included. The main contribution includes:
	\begin{itemize}
		\item A denotational semantics for nondeterministic quantum programs, which resembles the lifted model in the probabilistic setting. The rationale behind this design decision is an observation that the relational model cannot be satisfactorily defined in the quantum setting, partly due to the non-uniqueness of ensembles of pure states associated with a given density operator. 
		\item The partial and total correctness of nondeterministic quantum programs based on their denotational semantics. Compared with~\cite{li2014termination,li2017algorithmic}, our notions of correctness are much broader, capable of describing a wider range of properties beyond termination of quantum programs.	
		\item Hoare-type logic systems which are both sound and relatively complete with respect to partial/total correctness. Illustrating examples are explored and a prototyping tool is implemented which show the utility of these logic systems in proving the correctness of some quantum algorithms such as quantum error correction scheme, Deutsch algorithm, and a nondeterministic quantum walk.
\end{itemize}

%
%	Introducing nondeterminism to quantum programming languages has at least the following benefits: (1) to describe adversary opponent in a game setting or noisy channels in quantum communication; (2) to model user defined oracles which is called by the main program as a black box; (3) to model parallelism and concurrency in distributed quantum computing; (4) to allow underspecified programs and support program refinement process.	

	The remaining part of the paper is organised as follows. 
	In Sec.~\ref{sec:pre} we review some basic notions from quantum computing that will be used in later discussion. 
	The syntax and denotational semantics of  nondeterministic quantum programs considered in this paper are given in Sec.~\ref{sec:lang}. Sec.~\ref{sec:verification} is the main part of the paper where we elaborate the notions of assertions for quantum states, and the partial and total correctness of nondeterministic quantum programs. Proof systems for these two correctness notions are proposed and shown to be both sound and relatively complete.
	%, with  help of the weakest (liberal) preconditions. 
	Illustrative examples are explored in Sec.~\ref{sec:case} to show the effectiveness of our proof systems. Sec.~\ref{sec:implementation} is devoted to a proof assistant prototype implementing the verification techniques developed in this paper. Finally, Sec.~\ref{sec:conclusion} concludes the paper and points out some topics for future studies.

\section{Background on quantum computing}
\label{sec:pre}

This section is devoted to fixing some notations from linear algebra and quantum mechanics that will be used in this paper. For a thorough introduction of relevant backgrounds, we refer to~\cite[Chapter 2]{nielsen2002quantum}.

Let $\h$ be a finite-dimensional Hilbert space. Following the tradition of quantum computing, vectors in $\h$ are denoted in the Dirac form $|\psi\>$. The inner and outer products of two vectors $|\psi\>$ and $|\phi\>$ are written as $\<\psi|\phi\>$ and $|\psi\>\<\phi|$ respectively. Let $\lh$ be the set of linear operators on $\h$. Denote by $\tr(A)= \sum_{i\in I} \<\psi_i|A|\psi_i\>$ the \emph{trace} of $A\in \lh$ where $\{|\psi_i\> : i\in I\}$ is an (arbitrary) orthonormal basis of $\h$.  The \emph{adjoint} of $A$, denoted $A^\dag$, is the unique linear operator in $\lh$ such that $\<\psi|A|\phi\> = \<\phi|A^\dag |\psi\>^*$ for all $|\psi\>, |\phi\>\in \h$. Here, for a complex number $z$, $z^*$ denotes its conjugate. An operator $A\in \lh$ is said to be (1) \emph{hermitian} if $A^\dag = A$; (2) \emph{unitary} if $A^\dag A = I_\h$, the identity operator on $\h$; (3) \emph{positive} if for all $|\psi\>\in \h$, $\<\psi|A|\psi\>\geq 0$. 
Every hermitian operator $A$ has a \emph{spectral decomposition} form $A  = \sum_{i\in I} \lambda_i |\psi_i\>\<\psi_i|$ where $\{|\psi_i\> : i\in I\}$ constitute an orthonormal basis of $\h$. 
The L\"owner (partial) order $\le$ on $\lh$ is defined by letting $A\le B$ iff $B-A$ is positive. 

Let $\h_1$ and $\h_2$ be two finite dimensional Hilbert spaces, and $\h_1\otimes \h_2$ their tensor product. 
Let $A_i\in \l(\h_i)$. The tensor product of $A_1$ and $A_2$, denoted $A_1\otimes A_2$ is a linear operator in $\l(\h_1\otimes \h_2)$ such that
$(A_1\otimes A_2)(|\psi_1\>\otimes |\psi_2\>) = (A_1|\psi_1\>)\otimes (A_2|\psi_2\>)$ for all $|\psi_i\> \in \h_i$. The definition extends linearly to $\l(\h_1\otimes \h_2)$. To simplify notations, we often write $|\psi_1\> |\psi_2\>$ for $|\psi_1\>\otimes |\psi_2\>$.

A linear operator $\e$ from $\l(\h_1)$ to $\l(\h_2)$ is called a \emph{super-operator}.  It is said to be (1) \emph{positive} if it maps positive operators on $\h_1$ to positive operators on $\h_2$; (2) \emph{completely positive} if $\mathcal{I}_\h\otimes \e$ is positive for all finite dimensional Hilbert space $\h$, where $\mathcal{I}_\h$ is the identity super-operator on $\lh$; (3) \emph{trace-preserving} (resp. \emph{trace-nonincreasing}) if 
$\tr(\e(A)) = \tr(A)$ (resp. $\tr(\e(A)) \leq \tr(A)$ for any positive operator $A\in \l(\h_1)$.

From \emph{Kraus representation theorem}~\cite{kraus1983states}, a super-operator $\e$  from $\l(\h_1)$ to $\l(\h_2)$ is completely positive iff there are linear operators $\{E_i : i\in I\}$, called \emph{Kraus operators} of $\e$, from $\h_1$ to $\h_2$ such that $\e(A) = \sum_{i\in I} E_i A E_i^\dag$ for all $A\in \l(\h_1)$. Furthermore, $\e$ is trace-preserving (resp. {trace-nonincreasing}) iff $\sum_{i\in I} E_i^\dag E_i = I_{\h_1}$ (resp. $\sum_{i\in I} E_i^\dag E_i \le I_{\h_1}$). The \emph{adjoint} of $\e$, denoted $\e^\dag$, is a completely positive super-operator from $\l(\h_2)$ back to $\l(\h_1)$ with Kraus operators being $\{E_i^\dag : i\in I\}$. It is easy to check that for any $A\in \l(\h_1)$ and $B\in \l(\h_2)$, $\tr(\e(A)\cdot B) = \tr(A\cdot \e^\dag(B))$.

In the following, for simplicity, all super-operators are assumed to be completely positive and trace-nonincreasing unless otherwise stated. Given the tensor product space $\h_1\otimes \h_2$, the \emph{partial trace} with respect to $\h_2$, denoted $\tr_{\h_2}$, is a super-operator from $\l(\h_1\otimes \h_2)$ to $\l(\h_1)$ such that for any $|\psi_i\>, |\phi_i\> \in \h_i$, $i=1,2$,
$$\tr_{\h_2}(|\psi_1\>\<\phi_1|\otimes |\psi_2\>\<\phi_2|) = 
\<\phi_2|\psi_2\> |\psi_1\>\<\phi_1|.$$
Again, the definition extends linearly to $\l(\h_1\otimes \h_2)$.

According to von Neumann's formalism of quantum mechanics
\cite{vN55}, any quantum system with finite degrees of freedom is associated with a finite-dimensional Hilbert space $\h$ called its \emph{state space}. When the dimension of $\h$ is $2$, such a system is called a \emph{qubit}, the analogy of bit in classical computing. A {\it pure state} of the system is described by a normalised vector in $\h$. When the system is in state $|\psi_i\>$ with probability $p_i$, $i\in I$, it is said to be in a \emph{mixed} state, represented by the \emph{density operator} $\sum_{i\in I} p_i|\psi_i\>\<\psi_i|$ on $\h$. Obviously, a density operator is positive and has trace 1. 
In this paper, we follow Selinger's convention~\cite{selinger2004towards} to  regard \emph{partial density operators}, i.e. positive operators with traces not greater than 1 as (unnormalised) quantum states. Intuitively, a partial density operator $\rho$ denotes a legitimate quantum state $\rho/\tr(\rho)$ which is obtained with probability $\tr(\rho)$. Denote by $\dh$ the set of partial density operators on $\h$. The state space of a composite system (e.g., a quantum system consisting of multiple qubits) is the tensor product of the state spaces
of its components. For any $\rho$ in $\d(\h_1 \otimes \h_2)$, the partial traces
$\tr_{\h_1}(\rho)$ and $\tr_{\h_2}(\rho)$ are
the reduced quantum states of $\rho$ on $\h_2$ and $\h_1$, respectively. 

The \emph{evolution} of a closed quantum system is described by a unitary
operator on its state space: if the states of the system at 
$t_1$ and $t_2$ are $\rho_1$ and $\rho_2$, respectively, then
$\rho_2=U\rho_1U^{\dag}$ for some unitary $U$. Typical unitary operators used throughout this paper include Pauli-$X$, $Y$, and $Z$, Hadamard $H$, and controlled-NOT (CNOT) operator $CX$, represented in the matrix form (with respect to the computational basis) respectively as follows:
\[
X\define \begin{bmatrix}
	0 & 1 \\
	1 & 0 
\end{bmatrix},\ Y\define \begin{bmatrix}
	0 & -i \\
	i & 0 
\end{bmatrix},\ Z\define \begin{bmatrix}
	1 & 0 \\
	0 & -1
\end{bmatrix},
\]
and
\[H\define \frac{1}{\sqrt{2}} \begin{bmatrix}
	1 & 1 \\
	1 & -1
\end{bmatrix},\ CX\define \begin{bmatrix}
	1 & 0 & 0 & 0 \\
	0 & 1 & 0 & 0 \\
	0 & 0 & 0 & 1 \\
	0 & 0 & 1 & 0
\end{bmatrix}.\]
A (projective) quantum {\it measurement} $\m$ is described by a
collection $\{P_i : i\in O\}$ of projectors (hermitian operators with eigenvalues being either 0 or 1) in the state space $\h$, where $O$ is the set of measurement outcomes. It is required that the
measurement operators $P_i$'s satisfy the completeness equation
$\sum_{i\in O}P_i = I_\h$. If the system was in state $\rho$ before measurement, then the probability of observing outcome $i$ is given by
$p_i=\tr(P_i\rho),$ and the state of the post-measurement system
becomes $\rho_i = P_i\rho P_i/p_i$ whenever $p_i>0$. Sometimes we use 
a hermitian operator $M$ in $\lh$ called \emph{observable} to represent a  projective measurement. To be specific, let 
\[
M=\sum_{m\in \mathit{spec}(M)}mP_m
\] 
where $\mathit{spec}(M)$ is the set of eigenvalues of $M$, and $P_m$ the projector onto the eigenspace associated with $m$. Then the projective measurement determined by $M$ is $\{P_m : m\in \mathit{spec}(M)\}$. Note that by the linearity of the trace function,
the expected value of outcomes when $M$ is measured on state $\rho$ is calculated as
\[
\sum_{m\in \mathit{spec}(M)} m\cdot  \tr(P_m \rho) = \tr(M\rho).
\]

Finally, the dynamics that can occur in a (not necessarily closed) physical system are described by a trace-preserving super-operator. Typical examples include the unitary transformation $\e_U(\rho)\define U\rho U^\dag$ and the state transformation caused by a measurement, when all the post-measurement states are taken into account. More specifically, the evolution
\[
\e_\m(\rho) \define \sum_{i\in O} p_i \rho_i = \sum_{i\in O} P_i \rho P_i 
\]
is a super-operator. 

\textbf{Notation conventions.}  To simplify notations, we assume the following conventions throughout the paper. 
\begin{itemize}
	\item For a pure state $|\psi\>$ in a finite dimensional Hilbert space $\h$, we denote by $[|\psi\>]$ its corresponding (rank-1) density operator; that is, $[|\psi\>] \define |\psi\>\<\psi|$. 
	\item We use subscripts to indicate the quantum systems on which a state, an operator, or a super-operator is acting. For example, $|\psi\>_q$ or $[|\psi\>_q]$ denotes a pure state $|\psi\>$ of qubit $q$, and $CX_{q_1, q_2}$ denotes the CNOT operator with $q_1$ being its control qubit and $q_2$ the target qubit. Here $CX|x\>|y\> = |x\>|x\oplus y\>$ for any $x,y\in \{0,1\}$ and $\oplus$ denotes exclusive-or. Furthermore, for a $d$-dimensional Hilbert space, denote by $\{|i\> : 0\leq i\leq d-1\}$ its standard (or computational) basis.  
	\item Any super-operator is assumed to be identical to its cylinder extensions (the tensor product of it with the identity super-operator on the remaining subsystems) on larger Hilbert spaces. In other words, $\e$ from $\l(\h_1)$ to $\l(\h_2)$ can be regarded as from $\l(\h\otimes \h_1)$ to $\l(\h\otimes \h_2)$ for any $\h$ by identifying $\e$ and $\id_\h\otimes\e$. In particular, any value $p\in [0,1]$
	(a super-operator on the 0-dimensional space) can be regarded as $p\cdot \id_V$ on $\d(\h_{V})$ for any finite set $V$ of qubits.
	\item Operations on individual elements are assumed to be extended to sets in an element-wise way. For example, let $\supoprset$ and $\mathbb{F}$ be two sets of super-operators. Then $$\supoprset \circ \e + \mathbb{F}\circ \f\define \{\e' \circ \e+ \f' \circ \f : \e'\in \supoprset, \f'\in \mathbb{F}\}$$
	where $\circ$ denotes the composition of super-operators; that is,
	$\e\circ \f(\rho) = \e(\f(\rho))$ for all $\rho$. 
\end{itemize} 
	
	\section{Nondeterministic quantum programs}\label{sec:lang}
	
	We extend the purely quantum while-language defined in~\cite{feng2007proof,ying2012floyd} to describe quantum programs with nondeterministic choices.  Let $\QVar$, ranged over by $q, r, \cdots$, be a finite set of (qubit-type) quantum variables. For any
	 subset $V$ of $\QVar$, let
	\[\h_V \define \bigotimes_{q\in V} \h_{q},
	\]
	where $\h_{q} \simeq \h_2$ is the 2-dimensional Hilbert space associated with $q$. As we use subscripts to distinguish Hilbert spaces associated with different quantum variables, their order in the tensor product is irrelevant. 
	%In this paper, when we refer to a subset of $\QVar$, it is always assumed to be finite.
	
	\subsection{Syntax}
	A nondeterministic quantum program is defined by the following rules:
	\begin{align*} S::= &\ \sskip\ |\ \abort\ |\ \bar{q}:=0\ |\
		\bar{q}\apply U\ |\ S_0;S_1\ |\ S_0\ \square\ S_1\ |\\
		&\ \pmstm \ |\ \pwstm
	\end{align*}
	where $S,S_0$ and $S_1$ are nondeterministic quantum programs, $\bar{q} \define q_1, \ldots, q_n$ an (ordered) tuple of distinct qubit-type variables, and $U$ a unitary operator and $\m=\{P_0, P_1\}$ a two-outcome projective measurement on
	$\h_{\bar{q}}$ respectively.
	Sometimes we also use $\bar{q}$ to denote the (unordered) set $\{q_1,q_2,\dots,q_n\}$. Let $\qVar(S)$ be the set of quantum variables in $S$.
	
	The program constructs are standard, and their meaning will be clear when the denotational semantics is given later. Intuitively, $\sskip$ is a no-op statement, $\abort$  halts the computation with no proper state reached,
	$\bar{q}:=0$ initialises each qubit in system $\bar{q}$ into $|0\>$, $\bar{q}\apply U$ applies the unitary operator $U$ on system $\bar{q}$, $S_0;S_1$ is the sequential composition of $S_0$ and $S_1$, and $S_0\ \square\ S_1$ chooses $S_0$ or $S_1$ to execute non-deterministically. The conditional statement $\pmstm$ and the while statement $\pwstm$ behave similarly to their classical counterparts, except that they use the outcome of measuring $\m$ on $\bar{q}$ to determine subsequent operations. These quantum measurements often change the state of the system they are applied on. This is in contrast with classical programs, where such side-effects do not exist. 
	
			\begin{figure}[t]\centering
		\tikzset{
			my label/.append style={above right,xshift=0.3cm}
		}
		\begin{quantikz}[row sep=0.3cm,column sep=0.6cm]
			\lstick{$q\!:\  |\psi\>$} &\ctrl{1} & \ctrl{2}  &  %\gate[3]{\e_{\mathit{bf}}} 
			\gate[3,style={starburst,line
				width=1pt,inner xsep=-4pt,inner ysep=-5pt},
			label style=black]{\text{noise}} 
			& \ctrl{2} & \ctrl{1} & \gate{X} & \qw\\	  
			\lstick{$q_1\!:\  |0\>$} &\targ{}  &   \qw &   &\qw & \targ{} & \meter{} \vcw{-1}& \\
			\lstick{$q_2\!:\  |0\>$}  &\qw &\targ{}& &    \targ{}&\qw  &\meter{} \vcw{-1}& 
		\end{quantikz}
		\caption{Quantum error correction scheme which corrects a bit-flip error on any of the three qubits.}\label{fig-teleportation}
	\end{figure}
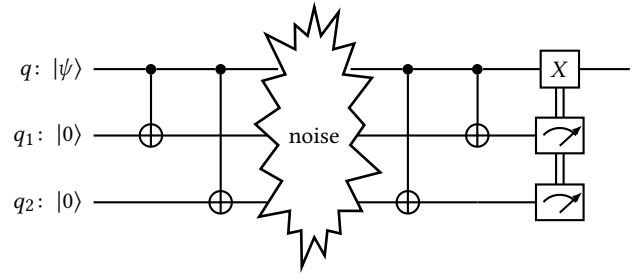
	\begin{example}[Quantum Error Correction Scheme as a Nondeterministic Program]\label{exm:errcor}
		Error correction codes are widely used to protect quantum information from noise in quantum communication channels or during quantum computation. The simplest error correction code, called three-qubit bit-flip code, encodes each qubit state $\alpha_0|0\> + \alpha_1|1\>$ into a three-qubit state $\alpha_0|000\> + \alpha_1|111\>$ (see the left part of Fig.~\ref{fig-teleportation} for a circuit implementation of the encoding process). All the three qubits then pass through a noisy quantum channel, which flips the qubit (that is, applies a Pauli-$X$ operator on it which turns $|0\>$ into $|1\>$ and $|1\>$ into $|0\>$ simultaneously) with some small probability.
		For simplicity, we assume that at most one of the three qubits is flipped, but we do not know which one it is. Interestingly, by applying a properly designed error correction procedure (shown in the right part of Fig.~\ref{fig-teleportation}), the error can be detected and corrected perfectly.     

If we model the unknown noise with a nondeterministic choice,
the quantum error correction scheme, including the processes of encoding, error introduction, and decoding, can be written as a nondeterministic quantum program as follows:
\begin{align*}
	\mathit{ErrCorr}\define &\\
	& q_1, q_2:=0;\\
	& q,q_1\apply \textit{CX};\ q,q_2\apply \textit{CX};\\
	& \sskip\ \square\ q\apply X\ \square\ q_1\apply X\ \square\ q_2\apply X;\\
	& q,q_2\apply \textit{CX};\ q,q_1\apply \textit{CX};\\
	& \iif\ \m[q_2]\ \then\\
	& \quad \iif\ \m[q_1]\ \then\\
	&\qquad q\apply X\\
	&\quad \pend\\ 
	&\pend
\end{align*}	
Here $CX$ is the CNOT gate, $\m$ is the measurement according to the computational basis $\{|0\>, |1\>\}$, and the command $\iif\ \m[r]\ \then\ S\ \pend$ denotes the abbreviation for $\iif\ \m[r]\ \then\ S\ \eelse\ \sskip\ \pend$. The nondeterministic statement models four different possibilities of error occurring: no error, bit-flip error on the first, second, and third qubit, respectively. Similar to the sequential composition, we assume (and will justify after the formal semantics is given) that the nondeterministic choice is both left- and right-associative.
\end{example}

	\subsection{Denotational Semantics}\label{sec:denotation}
	
	This subsection is devoted to a denotational semantics for nondeterministic quantum programs which is obtained by lifting the semantics of deterministic programs presented in~\cite{feng2007proof,ying2012floyd}. 
	
	\begin{figure}[t]
		\begin{align*}
			\sem{\sskip} & = \left\{1\right\}\\ 
			\sem{\abort} &=\left\{0\right\}\\
			\sem{\bar{q}:=0} &= \left\{\mathit{Set}^0_{\bar{q}}\right\} \\ 
			\sem{\bar{q}\apply U} &= \left\{\u_{\bar{q}}\right\}\\
			\sem{S_0;S_1} &= \sem{S_1}\circ \sem{S_0}\\ 
			\sem{S_0\ \square\ S_1} &= \sem{S_0} \cup \sem{S_1}\\
			\sem{\pmstm} &= \sem{S_0}\circ \p_{\bar{q}}^0 +  \sem{S_1}\circ \p_{\bar{q}}^1\\ \sem{\pwstm} &=
%			&\left\{\bigvee_{i\geq 0} \f_i^\eta : \eta\in \sem{S}^{\N}\right\}
		\end{align*}
		\[
		\left\{\sum_{i= 0}^\infty \p^0_{\bar{q}} \circ \eta_{i}\circ \p^1_{\bar{q}}\circ\ldots  \eta_1\circ \p^1_{\bar{q}} : \eta\in \sem{S}^{\N}\right\}
		\]
%		where $\f_{0}^\eta  \define 0 $, and $\forall\ i\geq 0$, 
%			\[
%			\f_{i+1}^\eta \define \f_{i}^\eta + \p^0_{\bar{q}} \circ \eta_{i}\circ \p^1_{\bar{q}}\circ\ldots  \eta_1\circ \p^1_{\bar{q}}.
%			\]
		\caption{Denotational semantics for nondeterministic quantum programs.}
		\label{fig:densemantics}
	\end{figure}
	
	%	Now we are able to define the denotational semantics of nondeterministic quantum programs using the operational one. 

	Given a finite dimensional Hilbert space $\h$, let $\s(\h)$ be the set of super-operators on $\h$. Note that $\s(\h)$ is a complete partial order (CPO)  with respect to the order $\preceq$ defined as follows: $\e \preceq \f$ iff there exists a super-operator $\g$ such that $\f = \e + \g$. The following lemma shows that this partial order coincides with the one defined by lifting pointwise the L\"{o}wner order over density operators on extended Hilbert spaces.
	\begin{lemma}
		For any $\e,\f\in \s(\h)$,  $\e \preceq \f$ iff for any auxiliary Hilbert space $\h'$ and any $\rho\in \d(\h\otimes \h')$,
		\[\e(\rho) \le \f(\rho).\]
	%	That is, $\f(\rho) - \e(\rho)$ is positive.
	\end{lemma}
	\begin{proof}
		The necessity part is easy. For the sufficiency part, note that the mapping $\g\define\f-\e$ is linear, trace non-increasing, and completely positive. Thus it is also a super-operator. 
	\end{proof}
	
	Let $\prog$ be the set of all nondeterministic quantum programs.
	For any $S\in \prog$, the \emph{denotational semantics} $\sem{S}$ of $S$ is a set of super-operators in $\s(\h_{\qVar(S)})$ defined inductively in Fig.~\ref{fig:densemantics}, where $\m = \{P^0, P^1\}$, $\p^i_{\bar{q}}$, $i=0,1$, is a super-operator such that $\p^i_{\bar{q}}(\rho) = P^i_{\bar{q}} \rho P^i_{\bar{q}}$, and $\mathit{Set}^0_{\bar{q}}$ and $\u_{\bar{q}}$ with $|\bar{q}|=n$ are super-operators such that $\mathit{Set}^0_{\bar{q}}(\rho) = \sum_{i=0}^{2^n-1}\quzi \rho\quiz$ and $\u_{\bar{q}}(\rho) = U_{\bar{q}} \rho U_{\bar{q}}^\dag$ for all $\rho\in \dhv$. Recall that the semantics of a deterministic quantum program is a single super-operator. We borrow the idea widely used in classical programming theory to employ sets of super-operators to represent nondeterminism in the semantics of nondeterministic quantum programs.
	
	A more elegant way to define the denotational semantics of nondeterministic quantum programs is to consider the power domain whose elements are subsets of super-operators with certain closure properties. The semantics of a while loop, say, is then simply the least fixed point of some Scott-continuous function over this power domain. However, since the main goal of this paper is to develop verification techniques for nondeterministic quantum programs, we decided to adopt the current definition, which is more intuitive and accessible for ordinary readers without a background in domain theory.  
	
	Note that with our notational convention, $\sem{S}$ can also be regarded as a subset of $\s(\h_{V})$ for any $V\supseteq \qVar(S)$. For example, the number $1$ in $\sem{\sskip}$ can be regarded as the identity super-operator $\id_{\h_V}$ for any $V\subseteq \QVar$. Similarly, the number 0 in $\sem{\abort}$ can represent the zero super-operator on any qubit system. Furthermore, note that $0(\rho)$ is the zero density operator for any input $\rho$. This accurately captures the intuition that $\abort$ applied to any input cannot produce any valid quantum state.
		 
	As the four basic commands $\sskip$, $\abort$, $\bar{q} := 0$, $\bar{q}\apply U$ are all deterministic, their semantic sets contain only a single super-operator representing the corresponding quantum operation. Recall also that we use the subscript $\bar{q}$ to denote the qubit system on which the quantum operation is applied. The semantics of $S_0;S_1$ is defined as the (element-wise) composition of $\sem{S_1}$ and $\sem{S_0}$. This follows the lifted model proposed in~\cite{jifeng1997probabilistic} for nondeterministic probabilistic programs.
	The reason why we adopt this lifted model (rather than the relational one) will be elaborated in Sec.~\ref{sec:comp}. 
				
	For $\sem{\pmstm}$, we have to combine the semantics of $S_0$ and $S_1$ using the corresponding measurement super-operators. Specifically, for any input state $\rho$, if the measurement returns 0, then the post-measurement state becomes $\p_{\bar{q}}^0(\rho)$, and one of the super-operators in $\sem{S_0}$ will be applied on this state. The case where the measurement result is 1 is similar. Finally, the (partial) density operators obtained from the two branches must be added together to form the output state of the entire statement, since the two branches are chosen probabilistically (due to the measurement), rather than non-deterministically.

	Finally, in $\sem{\pwstm}$, since the loop body $S$ may contain nondeterministic choices, we use a \emph{scheduler} $\eta$ to specify which super-operator in $\sem{S}$ is taken in each iteration of the while loop. To simplify the notation, we write $\eta_i$ for $\eta(i)$, the super-operator taken in the $i$-th iteration. For each $n\geq 0$, the super-operator
	\begin{equation}\label{eq:defF}
		\f_{n}^\eta \define \sum_{i= 0}^n \p^0_{\bar{q}} \circ \eta_{i}\circ \p^1_{\bar{q}}\circ\ldots  \eta_1\circ \p^1_{\bar{q}}
	\end{equation}
	is obtained from the first $n$ executions of the loop body $S$ under the scheduler $\eta$. It is evident that the sequence 
	$\f_{n}^\eta$, $n\geq 0$, is non-decreasing under $\preceq$. Thus the least upper bound $\bigvee_{n\geq 0} \f_{n}^\eta$ is well-defined, and it is the semantics of $\pwstm$ when $\eta$ is used to resolve the nondeterminism in $S$.

	\begin{example}\label{exa:dserrcor}
		We revisit the quantum error correction scheme presented in Example~\ref{exm:errcor}. First, it is easy to see that
		\begin{align*}
		&\sem{\mathit{ErrCorr}} =  \left\{\left(\x_q\circ \p_{q_1,q_2}^{11} + \p_{q_1,q_2}^{\neq 11}\right)\circ \mathcal{CX}_{q,q_1}\circ \mathcal{CX}_{q,q_2}\circ \u\right.\\
		&\ \ \left.\circ\ \mathcal{CX}_{q,q_2}\circ \mathcal{CX}_{q,q_1}\circ \mathit{Set}^0_{q_1,q_2} : \u \in \left\{\id, \x_q, \x_{q_1}, \x_{q_2}\right\} \right\}
		\end{align*}
	where $\mathcal{X}$ and $\mathcal{CX}$ are the super-operators corresponding to the unitary operators $X$ and $CX$, respectively. Furthermore, $$\p^{11}(\rho) = |11\>\<11|\rho|11\>\<11|$$ and $$\p^{\neq 11}(\rho) = \sum_{(i,j)\neq (1,1)}|ij\>\<ij|\rho |ij\>\<ij|$$ for all $\rho\in \d(\h_2 \otimes \h_2)$. That is, $\p^{11}_{q_1,q_2}$ denotes the super-operator corresponding to the case where the measurements $\m[q_2]$ and $\m[q_1]$ both obtain 1, while $\p^{\neq 11}_{q_1,q_2}$ represents the other cases. 
		
	From any input state $|\psi\>_q|\psi'\>_{q_1,q_2}$ where $|\psi\>=\alpha_0|0\> + \alpha_1|1\>$ and $|\psi'\> \in \h_2 \otimes \h_2$, 
	although there are four different super-operators in the denotational semantics $\sem{\mathit{ErrCorr}}$, when applied on $|\psi\>_q|\psi'\>_{q_1,q_2}$, the final quantum states $|\psi\>_q|ij\>_{q_1,q_2}$, $i,j=0,1$, have the same reduced state $|\psi\>$ on qubit $q$, which is exactly the input state on $q$. From this observation, we conclude that the error correction scheme has successfully corrected the potential bit-flip error occurred in any of the three qubits. 
	\end{example}

	The next lemma gives a recursive description of the denotational semantics of while loops. For any scheduler $\eta\in \sem{S}^\N$, we define $\eta^{\ra}\in \sem{S}^\N$ to be the suffix of $\eta$ starting from $\eta_2$; that is, $\eta^{\ra}_i = \eta_{i+1}$ for all $i\geq 1$.
	\begin{lemma}\label{lem:whilesem}
		Let $\while \define \pwstm$. Then
		for any $\eta\in  \sem{S}^\N$ and $n\geq 0$,
		\begin{equation}\label{eq:ra}
			\f_{n+1}^\eta  =  \p^0_{\bar{q}} + \f_n^{\eta^\ra} \circ \eta_1\circ \p^1_{\bar{q}}.
		\end{equation}
		where $\f_{n}^\eta$ is defined as in Eq.~\eqref{eq:defF}.
		Consequently,	
		\begin{equation}\label{eq:whilesem}
			\sem{\while} = \p^0_{\bar{q}} + \sem{\while} \circ \sem{S}\circ \p^1_{\bar{q}}.
		\end{equation}	
	\end{lemma}
	\begin{proof}
		Eq.~\eqref{eq:ra} can be easily proved by induction on $n$, and then the
		`$\subseteq$' part of Eq.~\eqref{eq:whilesem} follows.
		For the `$\supseteq$' part, let $\g\in \p^0_{\bar{q}} + \sem{\while} \circ \sem{S}\circ \p^1_{\bar{q}}$. Then there exists $\eta\in \sem{S}^\N$ and $\e\in \sem{S}$ such that 
		$
		\g = \p^0_{\bar{q}} + \bigvee_{n\geq 0} \f_n^{\eta} \circ \e\circ \p^1_{\bar{q}}.
		$
		Let $\eta'\in \sem{S}^\N$ with $\eta'_1 = \e$ and $\eta'_{i+1} = \eta_i$ for all $i\geq 1$. Then it is easy to prove by induction that for any $n\geq 1$, 
		\[
		\f_{n}^{\eta'} = \p^0_{\bar{q}} + \f_n^{\eta} \circ \e\circ \p^1_{\bar{q}}.
		\]
		Thus $\g = \bigvee_{n\geq 0} \f_n^{\eta'}$ is in $\sem{\while}$ as well.
	\end{proof}

	\subsection{Different Approaches for Semantics of Quantum Programs}
	\label{sec:comp}
	
	This subsection is devoted to an explanation of the design decisions we made in defining the semantics of nondeterministic quantum programs.
	
	\subsubsection{Pure-State v.s. Mixed-State Semantics}
	Note that there are two different approaches in defining semantics of deterministic quantum programs in the literature. One is based on pure states~\cite{van2004lambda,selinger2006lambda,sanders2000quantum,deng2015coinduction,pagani2014applying}, in which the meaning of a program is defined assuming that it is applied on pure states. The semantics is then extended to mixed states using, say, spectral decomposition. To be specific, if $\rho = \sum_{i\in I}p_i |\psi_i\>\<\psi_i|$ and a program $S$ maps $|\psi_i\>$ to $|\psi_i'\>$, then the final (mixed) state of executing $S$ on $\rho$ is  $\sum_{i\in I}p_i |\psi_i'\>\<\psi_i'|$. In contrast, the other approach is based on mixed states ~\cite{selinger2004towards,d2006quantum,feng2007proof,ying2012floyd}. That is, the semantics of a program is defined directly on mixed states without extension.

	The choice of pure-state v.s. mixed-state semantics does not make any difference when the programs considered are deterministic. The reason is that the semantics of a deterministic quantum program is a super-operator which is linear with respect to convex combination of density operators. Consequently, the two approaches indeed obtain the same semantics for deterministic quantum programs when applied on mixed states.	 
	However, this is not true for nondeterministic quantum programs, since now a program often corresponds to a set of super-operators, instead of a single one. As a result, convex combination of pure-state semantics does not give the same result as the mixed-state semantics. To make this point clear, let us see a simple example.
	
	\begin{example}\label{exm:pmsemantics}
		Let $S$ be a nondeterministic quantum program defined as follows:
		\begin{align*}
			S & \define \sskip\ \square\ q\apply X
		\end{align*}
		Then we have from Fig.~\ref{fig:densemantics} that $\sem{S} = \left\{1, \x_q\right\}$. Consequently,
		\begin{align}
			\sem{S}([|0\>]_q) & = \left\{[|0\>]_q, [|1\>]_q\right\} \notag\\
			\sem{S}([|1\>]_q) & = \left\{[|0\>]_q, [|1\>]_q\right\}  \label{eq:pmsemantics}\\
			\sem{S}([|+\>]_q) & = \left\{[|+\>]_q\right\} \notag\\
			\sem{S}([|-\>]_q) & = \left\{[|-\>]_q\right\} \notag
		\end{align}	   
		where $|\pm\> \define \frac{1}{\sqrt{2}}(|0\> \pm |1\>)$, and
		 the last equation follows from the fact that $[X|-\>]= [-|-\>]=[|-\>] $; that is, the global phase disappears in the density operator representation of quantum states.  
		
		Suppose we had adopted the pure-state approach to define our denotational semantics, and extended it to mixed states by taking the convex combination of pure-state semantics. Recall that in $\h_2$,
		\begin{equation}\label{eq:decompI}
			\frac{I}{2} = \frac{1}{2}\left([|0\>] + [|1\>]\right) = \frac{1}{2}\left([|+\>] + [|-\>]\right).
		\end{equation}
		Then from the first two equations in Eq.~\eqref{eq:pmsemantics} we would derive 
		\[
		\sem{S}(I_q/2) = \left\{[|0\>]_q, [|1\>]_q, I_q/2\right\},
		\]
		while from the last two equations in Eq.~\eqref{eq:pmsemantics} we would have
		$$
		\sem{S}(I_q/2) = \left\{I_q/2\right\}
		$$	   
		instead. 
		That is, the extended semantics for mixed states would not be well-defined.
		
		Note that this inconsistency does not exist for classical programs because, unlike the fact that a density operator can represent different ensembles of pure states, exemplified in Eq.~\eqref{eq:decompI} for $I/2$, any probability distribution of classical states has a unique representation (as a convex combination of underlying states).
		
	\end{example}

	\subsubsection{Relational v.s. Lifted Model}
	
	As pointed out in~\cite{jifeng1997probabilistic}, nondeterministic probabilistic programs can naturally be given two different semantic models: a relational one and a lifted one. Intuitively, in the lifted model, a nondeterministic program is semantically regarded as a set of deterministic  programs, while in the relational model, a dedicated semantics is given by taking into account the interference between probability and nondeterminism. The difference between these two models can be best illustrated in their treatment of sequential composition of programs. Specifically, let $\Sigma$ be the classical state space, ranged over by $s,t,$ etc. For simplicity, assume that $\Sigma$ is countable. Let $S$ and $T$ be nondeterministic probabilistic programs. Recall that the semantics $\sem{S}^{\mathbf{r}}$ of  $S$ in the relational model is a mapping from states in $\Sigma$ to sets of probability distributions over $\Sigma$. The semantics of the sequential composition $S;T$ is defined for any $s\in \Sigma$ as 
	\begin{equation}\label{eq:relational}
		\sem{S;T}^{\mathbf{r}}(s) = \left\{\sum_{t\in \Sigma} \mu(t) 	\cdot \nu_t : \mu\in \sem{S}^{\mathbf{r}}(s), \forall t. \nu_t\in \sem{T}^{\mathbf{r}}(t)\right\}.
	\end{equation}
	That is, each probability distribution in $\sem{S;T}^{\mathbf{r}}(s)$ is generated by first choosing a distribution $\mu$ in $\sem{S}^{\mathbf{r}}(s)$, then for each state $t$ choosing a distribution $\nu_t$ in $\sem{T}^{\mathbf{r}}(t)$, and finally taking the convex combination of $\nu_t$'s where the weight of $\nu_t$ is given by $\mu(t)$. In contrast, the semantics $\sem{S}^{\mathbf{l}}$ of  $S$ in the lifted model is a set of deterministic distribution-transformers which map states in $\Sigma$ to probability distributions over $\Sigma$. The lifted semantics of $S;T$ is defined to be	\begin{equation}\label{eq:lifted}
		\sem{S;T}^{\mathbf{l}} = \left\{ g\circ f : f\in \sem{S}^{\mathbf{l}}, g\in \sem{T}^{\mathbf{l}}\right\}
	\end{equation}			
	where $g\circ f$ is defined in the standard way: $(g\circ f)(s) = \sum_{t\in \Sigma} f(s)(t) 	\cdot g(t)$. Although these two models are both well-motivated, it was argued in~\cite{jifeng1997probabilistic} that the relational model is preferable since it enjoys more nice algebraic properties.
	
	Obviously, our definition in Fig.~\ref{fig:densemantics} follows the same style of the lifted model in Eq.~\eqref{eq:lifted}. To illustrate why it is problematic to have a relational semantics for nondeterministic quantum programs, let us take a close look at the definition of $\sem{S;T}^{\mathbf{r}}$ in Eq.~\eqref{eq:relational} from the perspective of a game. Intuitively, for each $t\in \Sigma$, $\nu_t$ is chosen from $\sem{T}^{\mathbf{r}}(t)$ by the adversary if the current program state is $t$. Note that this only makes sense under the  assumption that the program state at any given moment can be determined exactly during the runtime so that the adversary can use this information to make the choice. However, this assumption does not necessarily hold in the quantum setting, also due to the fact that a density operator can represent different ensembles of pure states. To make this point more clear, let us consider the following example.
	
		\begin{example}\label{exm:order}
		Let $S$ be defined as in Example~\ref{exm:pmsemantics} and		\begin{align*}
			T& \define q :=0; \ q\apply H;\ \mymeas\ q\\
			T_{\pm}& \define q :=0; \ \mymeas_{\pm}\ q
		\end{align*}
        \begin{sloppypar}
        	\noindent
		Here the command $\mymeas\ q$ is the abbreviation for $\iif\ \m_{0,1}[q]\ \then\ \sskip\ \eelse\ \sskip\ \pend$ and $\m_{0,1}$ is the measurement according to the computational basis $\{|0\>, |1\>\}$. Similarly, $\mymeas_{\pm}\ q$ is the abbreviation for $\iif\ \m_{\pm}[q]\ \then\ \sskip\ \eelse\ \sskip\ \pend$ and $\m_{\pm}$ is the measurement according to the orthonormal basis $\left\{|+\>, |-\>\right\}$. Note that both $T$ and $T_{\pm}$ are deterministic. Obviously, for any input state $\rho$ in $\d(\h_q)$ with $\tr(\rho) = 1$, the output state of $\ T$ is the ensemble $\left(|0\>: \frac{1}{2},  |1\> :  \frac{1}{2}\right)$, which can also be regarded as a probability distribution taking $|0\>$ or $|1\>$ uniformly. Similarly, the output state of $T_{\pm}$ for the same input $\rho$ is the ensemble $\left(|+\>: \frac{1}{2},  |-\> :  \frac{1}{2}\right)$, a probability distribution taking $|+\>$ or $|-\>$ uniformly. If we adopt the relational semantics $\sem{\cdot}^{\mathbf{r}}$ by regarding mixed quantum states as probability distributions over pure states, then we may have 
		\begin{align*}
			\sem{T}^{\mathbf{r}}(\rho) &= \left\{\left(|0\>: \frac{1}{2},  |1\> :  \frac{1}{2}\right)\right\},\\
			\sem{T_{\pm}}^{\mathbf{r}}(\rho) &= \left\{\left(|+\>: \frac{1}{2},  |-\> :  \frac{1}{2}\right)\right\}.
		\end{align*}
		Note that these two ensembles, although different in the probability distribution form, are physically indistinguishable; see Eq.~\eqref{eq:decompI}. Consequently, $\sem{T}^{\mathbf{r}}= \sem{T_{\pm}}^{\mathbf{r}}$.
		        \end{sloppypar}

		Now, consider the sequential composition of\; $T$ and $T_{\pm}$ with $S$ respectively. Similar to Eq.~\eqref{eq:relational}, we may have
		\begin{align*}
	\sem{T;S}^{\mathbf{r}}(\rho) &= \frac{1}{2} \cdot \sem{S}^{\mathbf{r}}\left([|0\>]\right) + \frac{1}{2} \cdot \sem{S}^{\mathbf{r}}\left([|1\>]\right) \\
	&= \frac{1}{2} \cdot \left\{|0\>, |1\>\right\} + \frac{1}{2} \cdot \left\{|0\>, |1\>\right\} \\
	& = \left\{
	\left(|0\>: 1\right), 
	\left(|1\>: 1\right), 
	\left(|0\>: \frac{1}{2},  |1\> :  \frac{1}{2}\right)
	\right\}
\end{align*}
	while
				\begin{align*}
			\sem{T_{\pm};S}^{\mathbf{r}}(\rho) &= \frac{1}{2} \cdot \sem{S}^{\mathbf{r}}\left([|+\>]\right) + \frac{1}{2} \cdot \sem{S}^{\mathbf{r}}\left([|-\>]\right)\\
			&= \frac{1}{2} \cdot \left\{|+\>\right\} + \frac{1}{2} \cdot \left\{|-\>\right\}  = \left\{
						\left(|+\>: \frac{1}{2},  |-\> :  \frac{1}{2}\right)
			\right\}.
		\end{align*}
        \begin{sloppypar}
		In the language of density operators, we can say that starting from $\rho$, 	$\sem{T;S}^{\mathbf{r}}$ may output $|0\>$, $|1\>$, or $\frac{I}{2}$ $\left(=  \frac{1}{2}|0\>\<0| +  \frac{1}{2}|1\>\<1|\right)$ non-deterministically while $\sem{T_{\pm};S}^{\mathbf{r}}$ can only output $\frac{I}{2}$ $\left(=  \frac{1}{2}|+\>\<+| +  \frac{1}{2}|-\>\<-|\right)$ although the nondeterministic choice also exists in it. In other words, $\sem{T;S}^{\mathbf{r}} \neq \sem{T_{\pm};S}^{\mathbf{r}}$ although $\sem{T}^{\mathbf{r}}= \sem{T_{\pm}}^{\mathbf{r}}$. 
		This is highly undesirable because semantic composability is a natural requirement of language design.
        \end{sloppypar}
	\end{example}

	\section{Verification of nondeterministic quantum programs}\label{sec:verification}
	
	Recall that a common practice in the verification of (deterministic or nondeterministic) probabilistic programs is to take expectations, that is, linear functions $f: \Sigma\ra [0,1]$ where $\Sigma$ is the set of classical states, as assertions. Then a program can be regarded as an expectation-transformer which maps a post-expectation to its greatest pre-expectation. The reason why deterministic and nondeterministic probabilistic programs can share the same form of assertions is that the set of expectations constitutes a complete lattice with respect to the pointwise partial order, where the join and meet are also defined pointwisely. Consequently, it is closed under taking infimum, the operation corresponding to the (demonic) nondeterminism.

	To describe desirable properties of a quantum state, we follow the approach of regarding hermitian operators from
		\[
	\p(\h_{\QVar}) \define \left\{M\in \l(\h_{\QVar}) : \z\le M\le I\right\}.
	\]
	 as quantum predicates in the verification of deterministic quantum programs~\cite{d2006quantum,feng2007proof,ying2012floyd}. Note that any hermitian operator $M\in \p(\h_{\QVar})$ induces a linear function $f_M: \d(\h_{\QVar})\ra [0,1]$ by setting $f_M(\rho) = \tr(M\rho)$ for any $\rho$.
	 Remarkably, recall from Sec.~\ref{sec:pre} that $\tr(M\rho)$ is exactly the expected value of outcomes when measuring the observable $M$ on state $\rho$. It is naturally interpreted as  the degree of $\rho$ satisfying $M$ if $M$ represents some desired property of the quantum system. 
	 	 
	 However, the set $\p(\h_{\QVar})$, although being a CPO, does not form a lattice. Thus it is not expressive enough for the verification of nondeterministic quantum programs. In this paper, we simply take $\a \define	2^{\p(\h_{\QVar})}$, the collection of subsets of $\p(\h_{\QVar})$, ranged over by $\qassert$, $\qassertp$, $\cdots$, as our set of quantum assertions. It is obviously a complete lattice in the usual subset order. When $\qassert = \{M\}$ is a singleton we simply write $\qassert$ as $M$. We further extend the operations applied on hermitian operators to quantum assertions in an element-wise way. For example, let $\e$ be a super-operator. Then by $\e(\qassert)$ we denote the set $\{\e(M) : M\in \qassert\}$.
	
\begin{definition}\label{def:satisfaction}
	Given a density operator $\rho\in \d(\h_{\QVar})$ and a quantum assertion $\qassert\in \a$, the expectation of $\qstate$ satisfying $\qassert$ is defined to be 
	$
	\Exp(\qstate \models \qassert) \define \inf_{M\in \qassert}
	\tr\left(M\rho\right).
	$ 
\end{definition}
The infimum taken in the above definition of $\Exp(\qstate \models \qassert)$ reflects a pessimistic view of the satisfaction: it provides a \emph{guaranteed} expected satisfaction of $\qassert$ by $\rho$ in the presence of possibly demonic nondeterminism. 
	
    \subsection{Correctness Formula}\label{sec:correctform}	
	As usual, program correctness is expressed by \emph{correctness formulas} with the form
	$\ass{\qassert}{S}{\qassertp}$
	where $S$ is a quantum program, and $\qassert$ and $\qassertp$ are  quantum assertions in $\a$.
	
	\begin{definition}\label{def:correctness}
		Let $S\in \prog$, and $\qassert, \qassertp\in \a$.
		\begin{enumerate}
			\item We say the correctness formula $\ass{\qassert}{S}{\qassertp}$ is true in the sense of \emph{total correctness}, written $\models_{\tot} \ass{\qassert}{S}{\qassertp}$, if for any 
			$\qstate\in \qstatesh{\QVar}$,
			$$	\Exp(\qstate \models \qassert)\leq \inf \left\{	\Exp(\sigma \models \qassertp) : \sigma\in \sem{S}(\qstate) \right\}.$$
			\item We say the correctness formula $\ass{\qassert}{S}{\qassertp}$ is true in the sense of \emph{partial correctness}, written $\models_{\pal} \ass{\qassert}{S}{\qassertp}$, if for any 
			$\qstate\in \qstatesh{\QVar}$,
			\begin{align*}
			    \Exp(\qstate \models \qassert)\leq &\inf \left\{	\Exp(\sigma \models \qassertp) + \tr(\rho) - \tr(\sigma)\right.\\
			    &\hspace{7em} \left.: \sigma\in \sem{S}(\qstate) \right\}.
			\end{align*}	
%			$$\inf\{\tr(M\rho): M\in \qassert\} \leq \inf \left\{\tr(N\sigma) + \tr(\rho) - \tr(\sigma): N\in \qassertp, \sigma\in \sem{S}(\qstate)\right\}.$$
		\end{enumerate}
	\end{definition}
	Intuitively, $\models_{\tot} \ass{\qassert}{S}{\qassertp}$ iff starting from any initial state, the guaranteed expected satisfaction of postcondition $\qassertp$ by any possible final state is lower bounded by the guaranteed expected satisfaction of precondition $\qassert$ by the initial state. For the case of partial correctness, we relax the lower bound by taking the non-termination probability $\tr(\rho) - \tr(\sigma)$ into account. It is easy to see that when $S$ is deterministic, and both $\qassert$ and $\qassertp$ contain only one quantum predicate, the above definition reduces to the corresponding one in~\cite{ying2012floyd} for deterministic quantum programs.
	
	\begin{example}\label{exa:err-correctness}
		The correctness of quantum error correction scheme for bit flip in Example~\ref{exm:errcor} can be stated as  follows: for any $|\psi\>\in \h_q$,
		\begin{equation}\label{eq:defexp}
			\models_{\tot}\ass{|\psi\>_q\<\psi|}{\mathit{ErrCorr}}{|\psi\>_{q}\<\psi|}.
		\end{equation}
		To see why Eq.\eqref{eq:defexp} captures the intuition that the (arbitrary)  state of qubit $q$ has been successfully
		protected from the possible bit-flip error, note that Eq.\eqref{eq:defexp} implies
		$$1=\<\psi|\psi\>\<\psi|\psi\> \leq \inf\left\{ \<\psi|\sigma|\psi\> : \sigma\in \sem{\mathit{ErrCorr}}(|\psi\>\<\psi|)\right\}.$$
		Thus $\sigma = |\psi\>\<\psi|$ for all $\sigma\in \sem{\mathit{ErrCorr}}(|\psi\>\<\psi|)$.
	\end{example}
	
		Note that for $*\in \{\tot, \pal\}$, if there exists $M\in \qassert$ such that $\models_* \ass{M}{S}{N}$ for all $N\in \qassertp$, then $\models_* \ass{\qassert}{S}{\qassertp}.$
	However, the reverse is not true. A counterexample is as follows. Let $\QVar \define \{q\}$, $\qassert \define \{|0\>_q\<0|, |1\>_q\<1|\}$, $\qassertp \define \{I_q/2\}$, and $S\define \sskip$. Note that for any 			$\qstate\in \qstatesh{\QVar}$, $\tr(|0\>_q\<0|\rho) + \tr(|1\>_q\<1|\rho) = \tr(\rho)$. Thus
	\begin{align*}
	\Exp\left(\qstate\models \qassert\right) &= \min \left\{\tr(|0\>_q\<0|\rho), \tr(|1\>_q\<1|\rho)\right\}\\
	&\leq \frac{\tr(\qstate)}{2} = \Exp\left(\qstate\models \qassertp\right),    
	\end{align*}
	and so $\models_{*} \ass{\qassert}{S}{\qassertp}$. However, neither 
	$\models_{*} \ass{|0\>_q\<0|}{S}{I_q/2}$ nor $\models_{*} \ass{|1\>_q\<1|}{S}{I_q/2}$ holds. To see the former one, we note that
	\[
	\Exp\left(|0\>_q\<0|\models |0\>_q\<0|\right) = 1 \not\leq 1/2  = \Exp\left(|0\>_q\<0|\models I_q/2\right).
	\] 
	For the latter one, we can take $|1\>_q\<1|$ as the initial state. 
	
	The following are some basic facts about total and partial correctness.
	
	\begin{lemma}\label{lem:corf}
		Let $S\in \prog$, $\qassert$ and $\qassertp$ be quantum assertions. 
		\begin{enumerate}
			\item If $\models_{\tot} \ass{\qassert}{S}{\qassertp}$ then $\models_{\pal} \ass{\qassert}{S}{\qassertp}$;
			\item $\models_{\tot} \ass{\z}{S}{\qassertp}$ and $\models_{\pal} \ass{\qassert}{S}{I}$;
			\item If $\models_{*} \ass{\qassert_i}{S}{\qassertp_i}$ for all $i$, then 
			$
			\models_{*} \ass{\bigcup_i \qassert_i}{S}{\bigcup_i \qassertp_i},
			$ where $*\in \{\tot, \pal\}$.
		\end{enumerate}
	\end{lemma}
	\begin{proof}
		Clause (1) follows directly from the fact that $\tr(\sigma)\leq \tr(\rho)$ for all $\sigma\in \sem{S}(\rho)$, clause (3) from the definition, and clause (2) from the observation that for any $M\in \qassert$, $\tr(M\rho) \leq \tr(\rho)$.
	\end{proof}	
	
	To conclude this subsection, we define a pre-order between quantum assertions by letting $\qassert \leinf \qassertp$ if for any $\rho$, $\inf_{M\in \qassert} \tr(M\rho) \leq \inf_{N\in \qassertp} \tr(N\rho)$. The following lemma is useful in our later discussion.
	
	\begin{lemma}\label{lem:propleinf}
		\begin{enumerate}
			\item Let $\e$ be a super-operator and $\qassert \leinf \qassertp$. Then $\e^\dag(\qassert) \leinf \e^\dag(\qassertp)$. Recall that $\e^\dag$ is the adjoint super-operator of $\e$. %In particular, for any program $S$, $xp.S.\qassert \leinf xp.S.\qassertp$. 
			\item If for all $i$, $\qassert_i \leinf \qassertp_i$, then $\bigcup_i \qassert_i \leinf \bigcup_i \qassertp_i$.
		\end{enumerate}
	\end{lemma}
 \begin{proof}
 	Easy from the definition of $\leinf$.
 \end{proof}

	\subsection{Proof Systems}
	
	{\renewcommand{\arraystretch}{2.7}
		\begin{figure}[t]
			\begin{lrbox}{\tablebox}
				\centering
				\begin{tabular}{l}
					\begin{tabular}{lc}
						(Skip)	& $\ass{\qassert}{\sskip}{\qassert}$\\  
						(Abort) 	& $\ass{I}{\abort}{\z}$\\
						\smallskip
						(Init)	& $\ass{\displaystyle\sum_{i=0}^{2^n-1} \quiz \qassert\quzi}{\bar{q}:=0}{\qassert}$ \\
						(Unit)	&
						$\ass{U_{\bar{q}}^\dag \qassert U_{\bar{q}}}{\bar{q}\apply U}{\qassert}$ \\
												\smallskip
						(Seq)	&
						$\displaystyle\frac{\ass{\qassert}{S_0}{\qassert'},\ \ass{\qassert'}{S_1}{\qassertp}}{\ass{\qassert}{S_0; S_1}{\qassertp}}$ \\
						(NDet)	&
						$\displaystyle\frac{\ass{\qassert}{S_0}{\qassertp},\ \ass{\qassert}{S_1}{\qassertp}}{\ass{\qassert}{S_0\ \square\ S_1}{\qassertp}}$\\
												\smallskip
						(Meas)	&
						$\displaystyle\frac{\ass{\qassert_1}{S_1}{\qassertp},\ \ass{\qassert_0}{S_0}{\qassertp}}{\ass{\p^0_{\bar{q}}(\qassert_0) + \p^1_{\bar{q}}(\qassert_1)}{\pmstm}{\qassertp}}$
						\\
						(While)	& $\displaystyle\frac{\ass{\qassert}{S}{\p^0_{\bar{q}}(\qassertp) + \p^1_{\bar{q}}(\qassert)}
							%\ \lambda_{max}(\p^0_{\bar{q}}(M)+ \p^1_{\bar{q}}(N)) \leq \lambda_{max}(M)
						}{\ass{\p^0_{\bar{q}}(\qassertp) + \p^1_{\bar{q}}(\qassert)}{\pwstm}{\qassertp}}$ \\
											\smallskip
						(Imp)	&
						$\displaystyle\frac{\qassert\leinf \qassert',\ \ass{\qassert'}{S}{\qassertp'},\ \qassertp'\leinf \qassertp}{\ass{\qassert}{S}{\qassertp}}$\\
						(Union)	&
						$\displaystyle\frac{ \ass{\qassert_i}{S}{\qassertp_i} \mbox{ for all } i\in I}{\ass{\bigcup_{i\in I}\qassert_i}{S}{\bigcup_{i\in I}\qassertp_i}}$
					\end{tabular}\\
				\end{tabular}
			\end{lrbox}
			\resizebox{0.5\textwidth}{!}{\usebox{\tablebox}}\\
			\vspace{4mm}
			\caption{Proof system for partial correctness. 
			}
			\label{tbl:psystem}
		\end{figure}
	}
	The core of Hoare logic is a proof system consisting of axioms and proof rules which enable syntax-oriented and modular reasoning of program correctness. In this section, we propose a Hoare logic for nondeterministic quantum programs. 
	
	\textbf{Partial Correctness}. We propose in Fig.~\ref{tbl:psystem} a proof system for partial correctness of quantum programs, which is a natural extension of the Hoare logic system introduced in~\cite{ying2012floyd} for deterministic quantum programs.   
	
	To help understand the rules presented in~Fig.~\ref{tbl:psystem}, let us compare them with their counterparts for classical programs. The rules (Skip), (Abort), (Seq), (NDet), (Imp) have the same form as the corresponding classical ones. Note that  for any pure state $|\psi\>$, $\tr(I\cdot |\psi\>\<\psi|) = 1$ and  $\tr(0\cdot |\psi\>\<\psi|) = 0$. Thus, the quantum predicates $I$ and $0$ play similar roles as $\true$ and $\false$ do respectively in classical assertions. The rules (Init) and (Unit) are the counterparts of the classical assignment rule, and they can be better understood in a backwards fashion. For example, (Unit) means that to guarantee postcondition $\qassertp$, it suffices to have $U_{\bar{q}}^\dag \qassert U_{\bar{q}}$ as the precondition.

	Recall the proof rule for classical conditional statements
	\begin{equation}\label{eq:crule} \displaystyle\frac{\ass{p\wedge B}{S_1}{q},\ \ass{p\wedge \neg B}{S_0}{q}
	}{\ass{p}{\iif\ B\ \then\ S_1\ \eelse\ S_0\ \pend}{q}}.
\end{equation}
	Due to the lack of conjunction of quantum assertions, our proof rule (Meas) takes an alternative form, adding the preconditions of the two branches of $\pmstm$ after applying the corresponding measurement super-operators. Note that $p = (p\wedge B) \vee (p\wedge \neg B)$. So (Meas) is indeed a generalisation of the rule in Eq.~\eqref{eq:crule}. Inspired by~\cite{ying2012floyd}, the quantum assertion $\p^0_{\bar{q}}(\qassertp) + \p^1_{\bar{q}}(\qassert)$ in rule (While) serves as a \emph{loop invariant} of the while loop. Finally, we introduce rule (Union) to combine different correctness formulas for the same quantum program, mimicking the conjunction rule for classical programs.
	
	 Denote by $\vdash_{\pal} \ass{\qassert}{S}{\qassertp}$ if the correctness formula $\ass{\qassert}{S}{\qassertp}$ can be derived from the proof system.

	\begin{theorem}\label{thm:psc}
		The proof system in Fig.~\ref{tbl:psystem} is both sound and relatively complete with respect to the partial correctness of nondeterministic quantum programs. 
	\end{theorem}
 	\begin{proof}(Sketch) Similar to~\cite{ying2012floyd}, the key idea for the proof here is to define the notion of weakest liberal precondition for quantum programs. Specially, for any nondeterministic quantum program $S$ and postcondition $\qassertp \in \a$, we construct a quantum assertion, denoted $wlp.S.\qassertp$, which is the weakest (or largest in terms of $\le_{\mathit{inf}}$) one among all valid preconditions. In other words, for any $\qassert\in \a$,  $\models_{\pal} \ass{\qassert}{S}{\qassertp}$ iff $\qassert\le_{\mathit{inf}} wlp.S.\qassertp$. 
 		
 	\begin{sloppypar}	
 	Now, to prove the soundness of our logic system, it suffices to show by induction on the structure of $S$ that whenever $\vdash_{\pal}\ass{\qassert}{S}{\qassertp}$, it holds $\qassert\le_{\mathit{inf}} wlp.S.\qassertp$. For the completeness part, we need to prove, again by structural induction, that $\vdash_{\pal}\ass{wlp.S.\qassertp}{S}{\qassertp}$. For more details, please refer to Appendix~\ref{app:proof}.  \qedhere
 	\end{sloppypar}
 	\end{proof}

	\textbf{Total Correctness}. Ranking functions play a central role in proving total correctness of while loop programs. Recall that in the classical case, a ranking function maps each reachable state during the execution of the loop body to an element of a well-founded ordered set (say, the set of nonnegative integers), such that the value decreases strictly after each iteration of the loop. Our proof rule for total correctness of while loops also heavily relies on the notion of ranking assertion. The following definition is inspired by the corresponding concept proposed in~\cite{feng2021quantum}. However, here we have to take into account possible nondeterministic choices in the loop body.	
	
	\begin{definition}\label{def:rank}
		\begin{sloppypar}
		Let $\widehat{\qassert}$ be a quantum assertion. A set  of quantum predicates $\left\{R_i^\eta: i\geq 0, \eta\in \sem{S}^{\N}\right\}$ is called a $\widehat{\qassert}$-\emph{ranking assertion} for $\pwstm$ if for each $\eta $, 
		\begin{enumerate}
			\item $\widehat{\qassert} \leinf R_0^\eta$;
			\item the sequence $R_i^\eta$, $i\geq 0$, is decreasing with respect to $\le$; that is, $R_0^\eta \ge R_1^\eta \ge\ldots$. Furthermore,  $\bigwedge_i R_i^\eta = \z$;
			\item for any $i\geq 0$  and $\qstate\in \d(\h_{\QVar})$,
			\begin{equation}\label{eq:rank}
				\tr\left( R_i^{\eta^\ra}\cdot \left(\eta_1\circ \p^1_{\bar{q}}(\qstate)\right)\right)\leq \tr\left(R_{i+1}^\eta\cdot \qstate\right).
			\end{equation}
			This can also be more compactly written as
	 \begin{equation}\label{eq:compact}
	 	\p^1_{\bar{q}} \circ \eta_1^\dag \left(R_i^{\eta^\ra}\right) \le R_{i+1}^\eta.
\end{equation}
		\end{enumerate}
\end{sloppypar}
	\end{definition}

	As can be seen from the rule (WhileT) below, $\widehat{\qassert}$ is usually taken to be the precondition of the correctness formula we are concerned with, which is also a loop invariant of $\pwstm$. The basic idea of 
	$\widehat{\qassert}$-{ranking assertion} is to provide a sequence of  upper bounds on $\widehat{\qassert}$. We will explain it in more detail below.
	
	For simplicity, assume $\widehat{\qassert} = \p^0_{\bar{q}}(M) + \p^1_{\bar{q}}(\qassert)$ for some quantum predicate $M$ and assertion $\qassert$. Let $M_0^\eta=\z$ and $M_i^\eta$, $i\geq 1$, be the quantum predicate representing the weakest precondition of $M$ if the loop body is only executed $i-1$ times; that is, for any $\rho$,
	\[
	\tr\left(M_i^\eta\cdot \rho\right) = \tr\left(M\cdot  \f_{i-1}^\eta(\rho)\right)
	\]
	where $\f_i^\eta$ is defined in Eq.~\eqref{eq:defF}. The goal of $\widehat{\qassert}$-{ranking assertion} $R_i^\eta$ is then to make sure that
	\begin{equation}\label{eq:rankbound}
		\p^0_{\bar{q}}(M) + \p^1_{\bar{q}}(\qassert)\le_{\mathit{inf}} M_i^\eta + R_i^\eta.
	\end{equation}
	Clause (1) of Definition~\ref{def:rank} establishes the bound for $i=0$, while clause (3), together with the assumption that $\p^0_{\bar{q}}(M) + \p^1_{\bar{q}}(\qassert)$ is a loop invariant, guarantee that Eq.~\eqref{eq:rankbound} holds inductively for all $i$. This can be shown by using Lemma~\ref{lem:whilesem}. Finally, because of clause (2), when $i$ tends to infinity, $M_i^\eta$ alone acts as the upper bound, which in turn implies that $\p^0_{\bar{q}}(M) + \p^1_{\bar{q}}(\qassert)$ is a valid precondition of $M$ with respect to $\pwstm$ in the total correctness sense.  
	
	With the notion of ranking assertion, we can state the rule (WhileT) for while loops in total correctness as follows:
	$$\small \mathrm{(WhileT)}\qquad  \displaystyle\frac{
		\begin{tabular}{c}
		$\ass{\qassert}{S}{\p^0_{\bar{q}}(\qassertp) + \p^1_{\bar{q}}(\qassert)}$\\
			$\p^0_{\bar{q}}(\qassertp) + \p^1_{\bar{q}}(\qassert)$-ranking assertion exists
		\end{tabular}
	}{\ass{\p^0_{\bar{q}}(\qassertp) + \p^1_{\bar{q}}(\qassert)}{\pwstm}{\qassertp}}$$
	The {proof system for total correctness} is then defined as for partial correctness, except that the rule (While) is replaced by (WhileT), and rule (Abort) replaced by
	\[\mathrm{(AbortT)}\qquad \ass{\z}{\abort}{\z}.\] 
	
	We write $\vdash_{\tot}\ass{\qassert}{S}{\qassertp}$ if the correctness formula $\ass{\qassert}{S}{\qassertp}$ can be derived using the proof system for total correctness. Again, we can prove the soundness and relative completeness of the proof system for total correctness.
	
	\begin{theorem}\label{thm:total}
		The proof system for total correctness is both sound and relatively complete with respect to the total correctness of nondeterministic quantum programs.
	\end{theorem}
	
	\begin{proof}(Sketch) Again, the key idea for the proof here is to define the notion of weakest precondition for quantum programs. The basic idea of using a ranking assertion to establish the soundness of rule (WhileT) has been intuitively discussed below Definition~\ref{def:rank}. For the completeness part, we can prove that the set of quantum predicates
\[
			R_i^\eta=
			\sum_{k=i}^\infty \p^1_{\bar{q}}\circ \eta_1^\dag\circ \ldots \circ \p^1_{\bar{q}}\circ \eta_k^\dag \circ \p^0_{\bar{q}}(I),
\]
representing the termination probability of the while loop after $i$ iterations of the loop body,
	constitute a proper ranking assertion to prove that $\vdash_{\tot} \ass{wp.\while.\qassert}{\while}{\qassert}.$
		 For more details, please refer to Appendix~\ref{app:proof}.  \qedhere
\end{proof}
	
	\section{Case studies}
	\label{sec:case}
	To illustrate the effectiveness of the proof systems
	 presented in the last section, we employ them to verify some simple quantum algorithms and protocols.
	
	\subsection{Three Qubit Quantum Error Correction Scheme}
	
	Recall from Example~\ref{exa:err-correctness} that the correctness of the three-qubit quantum error correction scheme can be stated as follows: for any $|\psi\> = \alpha_0|0\> + \alpha_1|1\>\in \h_2$,
	\begin{equation}\label{eq:eccor}
	\models_{\tot}\ass{|\psi\>_q\<\psi|}{\mathit{ErrCorr}}{|\psi\>_{q}\<\psi|}.
	\end{equation}
	First, we note that from rules (Skip) and (Unit),
	\begin{align*}
		&\vdash_{\tot}\ass{[\alpha_0|000\> + \alpha_1|111\>]}{\sskip}{[\alpha_0|000\> + \alpha_1|111\>]}\\
		&\vdash_{\tot}\ass{[\alpha_0|000\> + \alpha_1|111\>]}{q\apply X}{[\alpha_0|100\> + \alpha_1|011\>]}\\
		&\vdash_{\tot}\ass{[\alpha_0|000\> + \alpha_1|111\>]}{q_1\apply X}{[\alpha_0|010\> + \alpha_1|101\>]}\\
		&\vdash_{\tot}\ass{[\alpha_0|000\> + \alpha_1|111\>]}{q_2\apply X}{[\alpha_0|001\> + \alpha_1|110\>]}.
	\end{align*}
	Let $M_i$, $1\leq i\leq 4$, be the four postconditions presented above respectively. Then from (Imp), we have
	for each $S\in \{\sskip,\ q\apply X,\ q_1\apply X,\ q_2\apply X\}$,
\[
\ass{[\alpha_0|000\> + \alpha_1|111\>]_{q,q_1,q_2}}{S}{M_1+M_2+M_3+M_4}.\label{eq:eachs}
\]
	Thus we have the following proof outline:
		\begin{align*}
		&\left\{[\alpha_0|0\> + \alpha_1|1\>]_{q}\right\} & \\		
		& q_1, q_2:= 0;\\		
		&\left\{[\alpha_0|000\> + \alpha_1|100\>]_{q,q_1,q_2}\right\} & \mathit{(Init)}\\		
		& q,q_1\apply \textit{CX};\\
		&\left\{[\alpha_0|000\> + \alpha_1|110\>]_{q,q_1,q_2}\right\} & \mathit{(Unit)}\\
		& q,q_2\apply \textit{CX};\\
		&\left\{[\alpha_0|000\> + \alpha_1|111\>]_{q,q_1,q_2}\right\} & \mathit{(Unit)}\\
		&  \sskip\ \square\ q\apply X\ \square\ q_1\apply X\ \square\ q_2\apply X;\\
		&\left\{[\alpha_0|000\> + \alpha_1|111\>]  + [\alpha_0|001\> + \alpha_1|110\>]\right.&\\
		&\quad  \left.+  [\alpha_0|010\> + \alpha_1|101\>]  +  [\alpha_0|100\> + \alpha_1|011\>] \right\}& \mathit{(NDet)}\\
		& q,q_2\apply \textit{CX};\\
		&\left\{[\alpha_0|000\> + \alpha_1|110\>]  + [\alpha_0|001\> + \alpha_1|111\>]\right.&\\
		&\quad  \left.+  [\alpha_0|010\> + \alpha_1|100\>]  +  [\alpha_0|101\> + \alpha_1|011\>] \right\}& \mathit{(Unit)}\\
		& q,q_1\apply \textit{CX};\\
		&\left\{[\alpha_0|0\> + \alpha_1|1\>]_{q} \otimes ([|00\>] + [|01\>] + [|10\>])_{q_1,q_2}\right. \\
		&\quad \left. + [\alpha_0|1\> + \alpha_1|0\>]_{q} \otimes [|11\>]_{q_1,q_2} \right\}& \mathit{(Unit)}\\
			& \iif\ \m[q_2]\ \then\\
		&\quad\left\{[\alpha_0|0\> + \alpha_1|1\>]\otimes [|0\>] + [\alpha_0|1\> + \alpha_1|0\>]\otimes [|1\>]\right\} & \\			
		& \quad \iif\ \m[q_1]\ \then\\
		&\qquad\left\{[\alpha_0|1\> + \alpha_1|0\>]_{q}\right\} & \\		
		&\qquad q\apply X\\
		&\qquad\left\{[\alpha_0|0\> + \alpha_1|1\>]_{q}\right\} & \mathit{(Unit)}\\
		&\quad \pend\\ 
		&\quad\left\{[\alpha_0|0\> + \alpha_1|1\>]_{q}\right\} & \mathit{(Meas)}\\
		&\pend\\
		&\left\{[\alpha_0|0\> + \alpha_1|1\>]_{q}\right\} & \mathit{(Meas)}
	\end{align*}
	With the soundness of our logic system, this completes the proof of Eq.~\eqref{eq:eccor}.
	
	\subsection{Deutsch Algorithm}
	
	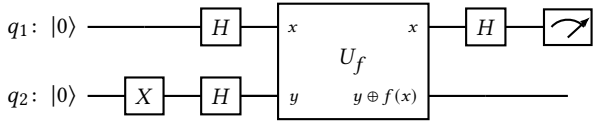
\begin{figure}[t]\centering
		\begin{quantikz}[row sep=0.3cm,column sep=0.5cm]
			\lstick{$q_1\!:\  |0\>$} &\qw &\gate{H} &\gate[wires=2][2cm]{U_f}
			\gateinput{$x$}	\gateoutput{$x$}&\gate{H} &\meter{} \\
			\lstick{$q_2\!:\  |0\>$}& \gate{X}&\gate{H}	&\qw\gateinput{$y$}\gateoutput{$y\oplus f(x)$}&\qw&\qw
		\end{quantikz}
		\caption{Quantum circuit for Deutsch algorithm.}\label{fig-deutsch}
	\end{figure}
	
	Deutsch algorithm~\cite{deutsch1985quantum} is one of the first quantum algorithms which demonstrate a speedup brought by quantum computing.
	Given a boolean function $f : \{0,1\} \rightarrow \{0,1\}$, Deutsch algorithm  can tell whether $f(0)$ equals $f(1)$ or not with just a single evaluation of $f$. In the traditional description of the algorithm, a quantum oracle $U_f$ that maps any state $|x\>\otimes|y\>$ to the state $|x\>\otimes |y\oplus f(x)\>$, where $x,y\in\{0,1\}$, is employed. The circuit for Deutsch algorithm is shown in Figure~\ref{fig-deutsch}. It claims that $f$ is constant  (meaning $f(0) = f(1)$) if the measurement outcome is 0; otherwise, it claims $f$ is balanced (meaning $f(0) \neq f(1)$).
	
	We now show how to describe Deutsch algorithm as a nondeterministic quantum program, and how to verify the correctness of it.  To this end, we note that
	\[
	U_f =
	\begin{cases}
		I_2 \otimes I_2 & \mbox{  if  } f(0)=f(1)=0;\\
		I_2\otimes X & \mbox{  if  } f(0)=f(1)=1;\\
		CX & \mbox{  if  }  f(0)=0 \mbox{ and } f(1)=1;\\
		C^0X & \mbox{  if  } f(0)=1 \mbox{ and } f(1)=0.\\
	\end{cases}
	\]
	where ${CX}$ is the CNOT gate, and ${C^0X}=(X\otimes I_2 )\cdot CX\cdot   (X \otimes I_2)$ which applies $X$ gate on the second qubit conditioning on the first qubit being $|0\>$. In the first two cases $f$ is constant while in the last two cases $f$ is balanced. 
	Then Deutsch algorithm can be written as 
	\begin{align*}
		\mathit{Deutsch}\define&\\
		& q_1,q_2:=0;\\
		& q_1\apply H;\ q_2 \apply X;\ q_2\apply H;\\
		& \iif\ \m_{0,1}[q]\ \then\\
		& \quad (q_1,q_2\apply {CX})\ \square \ (q_1,q_2\apply C^0X)\\
		& \eelse\\
		&\quad  \sskip\ \square\ (q_2\apply X)\\
		&\pend\\
		& q_1\apply H;\\
		&\mymeas\ q_1
	\end{align*}
	Here we introduce an auxiliary qubit $q$ with unknown initial state and use the measurement outcome on $q$ to choose $f$ (or equivalently, $U_f$). Note that for each outcome we have two possibilities for $U_f$, denoted by a nondeterministic choice between them. The statement $\mymeas\ q_1$ is defined as in Example~\ref{exm:order}.
	It is easy to prove from rules (Unit) and (Meas) that the following correctness formula is valid for any quantum assertion $\qassert$:
	\[
	\ass{|0\>_{q_1}\<0|\qassert|0\>_{q_1}\<0| + |1\>_{q_1}\<1|\qassert|1\>_{q_1}\<1|}{\mymeas\ q_1}{\qassert}.
	\]
	
	Note that at the end of the algorithm, both $q$ and $q_1$ are in one of the computational basis $\{|0\>, |1\>\}$. The correctness of Deutsch algorithm can then be stated as follows: 
	\begin{equation}\label{eq:deutcor}
		\models_{\tot}\ass{I}{\mathit{Deutsch}}{(|00\>\<00| + |11\>\<11|)_{q,q_1}}.
	\end{equation}
	That is, the (classical) information encoded in $q_1$ when the program terminates coincides with that encoded in $q$, thus indicating correctly whether $f(0)$ equals $f(1)$ or not.
	
	To prove Eq.~\eqref{eq:deutcor}, we first note from (Unit) that
	\begin{align*}
		&\vdash_{\tot}\ass{[|0\>_q|+-\>_{q_1,q_2}]}{\sskip}{[|0\>_q|+-\>_{q_1,q_2}]}\\
		&\vdash_{\tot}\ass{[|0\>_q|+-\>_{q_1,q_2}]}{q_2 \apply X}{[|0\>_q|+-\>_{q_1,q_2}]}\\
		&\vdash_{\tot}\ass{[|1\>_q|+-\>_{q_1,q_2}]}{q_1,q_2\apply CX}{[|1\>_q|--\>_{q_1,q_2}]}\\
		&\vdash_{\tot}\ass{[|1\>_q|+-\>_{q_1,q_2}]}{q_1,q_2\apply C^0X}{[|1\>_q|--\>_{q_1,q_2}]}
	\end{align*}
where $|\pm\> \define \frac{1}{\sqrt{2}}(|0\> \pm |1\>)$.
	Then we have the following proof outline:
	\begin{align*}
		& \left\{ I \right\} & \\
		& q_1,q_2:=0;\ \\
		& \left\{ [|00\>]_{q_1,q_2} \right\} &  \mathit{(Init), (Seq)}\\
		& q_1\apply H;\ q_2 \apply X;\ q_2\apply H;\\
		& \left\{ [|0+-\>]_{q,q_1,q_2} + [|1+-\>]_{q,q_1,q_2} \right\} &  \mathit{(Unit), (Seq)}\\
		& \iif\ \m_{0,1}[q]\ \then\\
		& \left\{ [|1+-\>]_{q,q_1,q_2}\right\} &\\
		& \quad (q_1,q_2\apply {CX})\ \square \ (q_1,q_2\apply C^0X)\\
		& \left\{ [|1--\>]_{q,q_1,q_2}\right\} & \mathit{(NDet)} \\
		& \eelse\\
		& \left\{ [|0+-\>]_{q,q_1,q_2}\right\} &\\
		&\quad  \sskip\ \square\ (q_2\apply X)\\
		& \left\{ [|0+-\>]_{q,q_1,q_2}\right\} & \mathit{(NDet)}\\
		&\pend\\
		& \left\{[|0+-\>]_{q,q_1,q_2} + [|1--\>]_{q,q_1,q_2}\right\} & \mathit{(Imp), (Meas)}\\
		& \left\{[|0+\>]_{q,q_1} + [|1-\>]_{q,q_1}\right\} & \mathit{(Imp)}\\
		& q_1\apply H;\\
		& \left\{[|00\>]_{q,q_1} + [|11\>]_{q,q_1}\right\} & \mathit{(Unit)}\\
		& \mymeas\ q_1 \\
		& \left\{[|00\>]_{q,q_1} + [|11\>]_{q,q_1}\right\} & \mathit{(Meas)}
	\end{align*}
	Finally, Eq.~\eqref{eq:deutcor} follows from the soundness of our logic system.
	
	\subsection{A Nondeterministic Quantum Walk} \label{exa:nondet-quantum-walk}
	
	To illustrate the utility of our logic system regarding quantum loops, let us consider a revised version of the nondeterministic quantum walk on a circle with four vertices presented in~\cite{li2014termination}. The walk uses a two-qubit system consisting of $q_1$ and $q_2$ as its principle system. It starts in the initial state $|00\>$, and at each step of walk, the following two unitary walk operators 
	\[ \scalebox{.9}{$
	W_1\define \frac{1}{\sqrt{3}}
	\begin{pmatrix}
		1 & 1 & 0 & -1\\
		1 & -1 & 1 & 0\\
		0 & 1 & 1 & 1\\
		1 & 0 & -1 & 1			
	\end{pmatrix},\quad
	W_2\define \frac{1}{\sqrt{3}}
	\begin{pmatrix}
		1 & 1 & 0 & 1\\
		-1 & 1 & -1 & 0\\
		0 & 1 & 1 & -1\\
		1 & 0 & -1 & -1			
	\end{pmatrix}$}
	\]
	are applied on $q_1$ and $q_2$ consecutively. Here $W_1$ and $W_2$ are written with respect to the computational basis $\{|00\>, |01\>, |10\>, |11\>\}$ of $\h_{q_1, q_2}$. However, the order in which these walk operators are applied is nondeterministically chosen.
	We further assume that there exists an absorbing boundary at the subspace spanned by $|10\>$; that is, after each walk step, a projective measurement $\m = \{P_0, P_1\}$ where
	$$P_0\define |10\>\<10|,\qquad P_1\define I_4 - |10\>\<10|$$ is applied.
	If the outcome corresponding to $P_0$ is observed, the whole process terminates; otherwise the next walk step continues.
	
	The walk can be described as the following nondeterministic program:
	\begin{align*}
	\mathit{QWalk}\define&\\
		& q_1,q_2:=0;\\
		& \while\ \m[q_1,q_2]\ \ddo\\
		&\quad \qquad  (q_1,q_2\apply W_1;\ q_1,q_2\apply W_2)\\ &\quad\ \square\quad (q_1,q_2\apply W_2;\ q_1,q_2\apply W_1)\\
		&\pend
	\end{align*}
	It has been shown in~\cite{li2014termination} that $\mathit{QWalk}$ does not terminate if the left branch of the nondeterministic choice is always taken in each iteration. This can be easily seen from the fact that $W_2W_1|00\> = |00\>$. In the following, we show a stronger result using our proof system for nondeterministic quantum programs: in fact, $\mathit{QWalk}$ does not terminate under any scheduler of the nondeterministic choice! For simplicity, we omit the subscript $\{q_1, q_2\}$ of the assertions and super-operators in the discussion below.
	
	First, note that this non-termination property can be stated as follows: 
	\begin{equation}\label{eq:corqwalk}
		\models_{\pal}\ass{I}{\mathit{QWalk}}{\z}.
	\end{equation}
	The reason is, from Definition~\ref{def:correctness}(2), Eq.~\eqref{eq:corqwalk} holds iff for any $\rho\in \d(\h_{q_1, q_2})$
	and $\sigma\in \sem{\mathit{QWalk}}(\rho)$,
	\[
	\tr(\rho) \leq \tr(\rho) - \tr(\sigma).
	\]
	Thus $\tr(\sigma)=0$, meaning that starting from any $\rho$, the probability of $\mathit{QWalk}$ reaching any valid quantum state is 0.
	
	To prove Eq.~\eqref{eq:corqwalk}, we first compute using rule (Unit) that
	{\small
	\begin{align*}
		&\left\{ [|00\>] +  \left[\frac{1}{\sqrt{2}}(|01\> + |11\>)\right] \right\} \\
		& q_1,q_2\apply W_1;\\
		&\left\{\left[\frac{1}{\sqrt{3}} (|00\> +  |01\> + |11\>)\right] + \left[\frac{1}{\sqrt{6}} (-|01\> +2  |10\> + |11\>)\right] \right\}\\
		& q_1,q_2\apply W_2;\\
		&\left\{ [|00\>] + \left[\frac{1}{\sqrt{2}}(|01\> + |11\>)\right]\right\} 
	\end{align*}
	}
	and 
	{\small
	\begin{align*}
		&\left\{ [|00\>] + \left[\frac{1}{\sqrt{2}}(|01\> + |11\>)\right]\right\} &\\
		&\left\{  \left[\frac{1}{3} (-|00\> + 2 |01\> + 2|11\>)\right]  + \left[\frac{1}{3\sqrt{2}}(4|00\> + |01\> + |11\>)\right]\right\}\\
		& q_1,q_2\apply W_2;\\
		&\left\{\left[\frac{1}{\sqrt{3}} (|00\> +  |01\> - |11\>)\right] + \left[\frac{1}{\sqrt{6}} (2|00\> -  |01\> + |11\>)\right] \right\}\\
		& q_1,q_2\apply W_1;\\
		&\left\{ [|00\>] + \left[\frac{1}{\sqrt{2}}(|01\> + |11\>)\right]\right\} 
	\end{align*}}
	Now let $N \define [|00\>] + \left[\frac{1}{\sqrt{2}}(|01\> + |11\>)\right]$, and $S$ the body of the $\while$ loop of $\mathit{QWalk}$. Note that $\p^0(\z)  = \z$ and $\p^1(N) = N$ where for $i=0,1$, $\p^i$ is the super-operator with a single Kraus operator $P_i$. Then by (NDet) we have
	\begin{equation}\label{eq:walkinv}
		\vdash_{\pal}\ass{N}{S}{\p^0(\z) + \p^1(N)}
	\end{equation}
	 Finally, we have 
	 	\begin{align*}
	 	&\left\{ I\right\} \\		
	 	&q_1, q_2:=0;\\
	 	&\left\{ [|00\>] + \left[\frac{1}{\sqrt{2}}(|01\> + |11\>)\right]\right\} \hspace{9em} \mathit{(Init)}\\
	 	& \while\ \m[q_1,q_2]\ \ddo\\
	 	& (q_1,q_2\apply W_1;\ q_1,q_2\apply W_2)\ \square\ (q_1,q_2\apply W_2;\ q_1,q_2\apply W_1)\\
	 	&\pend\\
	 	&\left\{\z\right\} \hspace{15.5em} \mathit{(While), Eq.~\eqref{eq:walkinv}}
	 \end{align*}
	With the soundness theorem~\ref{thm:psc}, this concludes the proof of Eq.~\eqref{eq:corqwalk}.

			\section{A prototype tool implementation}    \label{sec:implementation}

			We have illustrated the utility of our proof systems through a couple of examples in the previous section. It can be seen that even for such simple programs, carrying out formal verification is tedious (albeit routine) and involves a large number of matrix manipulations. To ease the burden on human users so that they can focus on more challenging parts such as specifying invariants for while loops, we implement a prototype proof assistant NQPV to automate routine parts of the verification process.
			Currently, NQPV only supports partial correctness; verification of total correctness is left as future work.

			 NQPV is a pure Python project that does not rely on existing theorem provers such as Isabelle and Coq. We make this decision taking into account the following factors: (1) Python and its various libraries provide
			 powerful matrix manipulation capabilities that facilitate us to compute the pre- and post-conditions of quantum programs and determine the $\le_{\mathit{inf}}$ relation between quantum assertions easily; (2) The main purpose of NQPV is to illustrate the practicality of our proposed logic system, so the ease of use is our top priority. Python's popularity and rich matrix libraries make it convenient and natural to describe quantum programs and assertions; (3) Symbolic verification in Python is difficult, which limits the quantum programs that NQPV can express and verify. But in return, most verification steps can be done automatically in NQPV.
			 
			 NQPV relies on the following Python packages: \texttt{ply} for syntax analysis, \texttt{numpy} for operator calculation, and \texttt{cvxpy} for the SDP solver. NQPV utilises a simple language to define and manage terms of operators and proofs. 
			
			\subsection{User Inputs}\label{sec:userinput}
			
			NQPV takes a correctness formula defined in Sec.~\ref{sec:correctform} as the input. However, to automate the proof for while programs, loop invariants are also expected whenever a while structure is encountered. For illustration, we reconsider the quantum walk example in Sec.~\ref{exa:nondet-quantum-walk}. The code is shown as follows:			
	\lstset{
		basicstyle=\tt,
		keywordstyle=\color{blue}\bfseries,
		%identifierstyle=\color{brown!80!black},
		morekeywords={def, load, end, proof, skip, show, while, do, inv, if, then, else},
        escapeinside={(*}{*)},
        morecomment=[l][\color{gray}]{//},
		emph={q,q1,q2},
		emphstyle={\color{black}},
		showstringspaces=false,
	}  
	\begin{lstlisting}  
def invN := load "invN.npy" end
// more operators imported ...
def pf := proof[q1 q2] :
    { I[q1] };
    [q1 q2] :=0;
    { inv: invN[q1 q2] };
    while MQWalk[q1 q2] do
        ( [q1 q2] *= W1; [q1 q2] *= W2
        # [q1 q2] *= W2; [q1 q2] *= W1 )
    end;
    { Zero[q1] }
end 
show pf end
	\end{lstlisting}
	Here the unitary operators \texttt{W1}, \texttt{W2} and measurement  \texttt{MQWalk} are defined in Sec.~\ref{exa:nondet-quantum-walk}, and they should be input by the user as \texttt{numpy} matrices. Note that some identifiers such as \texttt{I} and \texttt{Zero} are reserved for commonly used unitary operators, hermitian operators, and measurements.
    Furthermore, loop invariants (\texttt{invN} in this example) specified by the user for while programs should also be \texttt{numpy} matrices.
    
    In addition to the main program describing behaviours of the quantum walk, the above code also includes two extra lines \texttt{\{ I[q1] \}} and \texttt{\{ Zero[q1] \}} to represent the precondition and postcondition of the program. Hence, the body of proof term \texttt{pf} indeed describes the correctness formula presented in Eq.~\eqref{eq:corqwalk}. Finally, the keyword \texttt{inv} marks a loop invariant, and the \texttt{show} command in the last line outputs the proof outline generated by NQPV. 
    
    To further simplify usage, NQPV also allows users to omit preconditions and specify only postconditions. In this case, NQPV outputs the weakest precondition it can compute using the given postcondition and user-supplied loop invariants.
			
			\subsection{Verification Process}
			
			After successfully parsing the input, NQPV inductively constructs proofs according to the logic system  in Fig.~\ref{tbl:psystem} in an automated way. The strategy is to calculate the weakest preconditions in the backward direction, starting from the postcondition of the whole program. For while loops, NQPV will check if the loop invariant provided by the user is a valid one.

			In the end, the assistant compares the verification condition and the precondition proposed by the user (details will be given in the next subsection) and then generates the final result. 
			Again, the following shows the output of NQPV for the quantum walk example:

        \begin{lstlisting}
proof [q1 q2] : 
    { I[q1] };
    (*\textbf{\{ VAR2[q1 q2] \};}*) // the Veri. Con.
    [q1 q2] :=0;
    (*\textbf{\{ invN[q1 q2] \};}*)
    { inv: invN[q1 q2] };
    while MQWalk[q1 q2] do
        (*\textbf{\{ invN[q1 q2] \};}*)
        ( (*\textbf{\{ invN[q1 q2] \};}*) [q1 q2] *= W1;
          (*\textbf{\{ VAR0[q1 q2] \};}*) [q1 q2] *= W2
        # (*\textbf{\{ invN[q1 q2] \};}*) [q1 q2] *= W2;
          (*\textbf{\{ VAR1[q1 q2] \};}*) [q1 q2] *= W1 )
    end;
    { Zero[q1] }        
        \end{lstlisting}

    As is shown, every sub-program statement is annotated with the corresponding pre- and postconditions. Some of them are already defined, such as \texttt{invN}. Other operators, such as \texttt{VAR0}, are generated by NQPV to represent predicates used in the proof outline.
    %represented and stored by an automatic name. 
    %From the proof outline, we know that \texttt{invN} is exactly the weakest precondition (i.e. not only an approximation for the while statement), since it coincides with the loop invariant we have chosen. 
    The detailed information of these operators can be shown using the \texttt{show} command. For example, \texttt{show VAR2 end} for this example will return the two-qubit identity operator, meaning that the verification condition determined by the given postcondition and the loop invariant is \texttt{I[q1 q2]} (thus the original correctness formula is valid).

    In addition to checking the validity of a correctness formula, NQPV can also verify if a user-supplied operator is indeed a loop invariant. For example, if we change the \texttt{invN[q1 q2]} in the code in Sec.~\ref{sec:userinput} into \texttt{P0[q1]} where $\texttt{P0}=\ket{0}\bra{0}$, we will get an error message:
    \begin{lstlisting}
Error: 
Order relation not satisfied: 
    { P0[q1] } <= { VAR0[q1 q2] VAR3[q1 q2] }
...
Error: The predicate '{ P0[q1] }' is not 
    a valid loop invariant.
...
    \end{lstlisting}
           
		\subsection{Determining the $\le_{\mathit{inf}}$ Relation}
		
		One of the key steps in verifying a correctness formula is to determine if two assertions (sets of quantum predicates) satisfy the $\le_{\mathit{inf}}$ relation. For the sake of implementability, we assume all the quantum assertions allowed in NQPV are finite ones.
			
			\begin{lemma}\label{lem:inforder} Let $\qassert$ and $\qassertp$ are finite set of quantum predicates. Then $\qassert\le_{\mathit{inf}} \qassertp$ iff for any $N\in\qassertp$,
				$$
				\forall\rho\in\dhv, \exists M\in \qassert, \tr (M \rho) \leq \tr(N\rho).
				$$
			\end{lemma}
			 \begin{proof}
				From the finiteness of $\qassert$ and $\qassertp$, 
\begin{align*}
	&\qassert\le_{\mathit{inf}} \qassertp\\
	&\ \Leftrightarrow \ \forall\rho\in\dhv, \min_{M\in \qassert} \tr (M \rho) \leq \min_{N\in \qassertp} \tr(N\rho)\\
	&\ \Leftrightarrow \ \forall\rho\in\dhv, \forall N\in \qassertp, \exists M\in \qassert, \tr (M \rho) \leq \tr(N\rho)\\
	&\ \Leftrightarrow \ \forall N\in \qassertp, \forall\rho\in\dhv, \exists M\in \qassert, \tr (M \rho) \leq \tr(N\rho). \qedhere
\end{align*}
				 \end{proof}	
			With Lemma~\ref{lem:inforder}, we have the following algorithm to determine if $\qassert\le_{\mathit{inf}} \qassertp$ for finite sets $\qassert$ and $\qassertp$:
			\begin{itemize}
				\item If $\Theta = \{M\}$ is a singleton, then $\qassert\le_{\mathit{inf}} \qassertp$ iff $M\le N$ for all $N\in \qassertp$. This can be done by simply checking if the eigenvalues of $N-M$ are all nonnegative. 
				\item Otherwise, for each $N\in \qassertp$ we try to solve the following SDP problem:
				\begin{align*}
					&\mathrm{minimize}\ &0 \\
					&\mathrm{subject\ to}\ &\forall M \in \Theta, \tr(N\rho) \leq \tr(M\rho) -\epsilon \\
					& & \z \sqsubseteq \rho
				\end{align*}
				Here $\epsilon$ is a user-defined precision that is sufficiently small but positive. This is required by SDP solvers such as MOSEK and CVX. Another reason for the introduction of $\epsilon$ is to make the feasible region close and thus the whole problem an SDP one.  Obviously, if any of the $|\qassertp|$ SDP problems returns a solution, then $\qassert\not\le_{\mathit{inf}} \qassertp$. However, because of the precision parameter, we cannot make sure if $\qassert\le_{\mathit{inf}} \qassertp$ holds even none of the SDP problems returns a solution. Nevertheless, the probability of getting incorrect answers can be negligible if $\epsilon$ is taken sufficiently small.
			\end{itemize}

			\subsection{Related Tools}
			
	Program verification can also be done within an interactive theorem prover such as Isabelle and Coq. The first theorem prover for quantum programs is QHLProver~\cite{Liu.2019}, which implements the program logic in~\cite{ying2012floyd} using Isabelle/HOL. CoqQ~\cite{Zhou.2022} implements the same logic in Coq, with a general and abstract representation of linear operators. Qrhl-tool~\cite{Unruh.2019} is a verification tool embedded in Isabelle that supports relational verification of quantum programs.

        To make a fair comparison, we investigate the difference between NQPV and the existing tools in the following aspects (we only take CoqQ as an example, as QHLProver and Qrhl-tool are built in a similar way to CoqQ):
        \subsubsection*{Reliability} 
        For CoqQ, inference rules are not only defined, but their soundness with respect to the denotational semantics (also formally defined in CoqQ) is proved within Coq.	In contrast, the soundness of inference rules is not checked in NQPV. Nevertheless, a rigorous proof for the soundness of our logic system has been given in Theorem~\ref{thm:psc}. 
        \subsubsection*{Expressiveness}
        % Consider the problems that can be expressed and solved. 
        Existing tools only deal with deterministic quantum programs, while our prototype implementation supports nondeterministic ones. On the other hand, symbols can be handled in CoqQ, as in theorem provers; whereas only numerical terms are allowed in NQPV, as in Python. 
        In particular, CoqQ is suitable for verifying general (i.e., with an arbitrary number of qubits) algorithms such as Grover algorithm, since the qubit number appears as a symbol in CoqQ.
        \subsubsection*{Automation}
        Due to the abstract representation of linear operators in CoqQ, many properties of quantum predicates are not easy to prove and have to be shown by the user. However, NQPV can take advantage of powerful Python libraries, so many proofs can be automated without user assistance. 
        Therefore, the proof script for NQPV is usually much shorter than the proof written in CoqQ. 
        % For comparison, the proof of Grover search algorithm takes 180 lines in CoqQ, while the proof of the same algorithm for a particular instance in NQPV requires only 30 lines.
        \subsubsection*{Usability}
        Using the Coq library \texttt{MathComp}, CoqQ develops an abstract representation of linear operators directly based on Hilbert spaces, and a smart way to describe them using labelled Dirac notations. Consequently, users of CoqQ are supposed to have sufficient knowledge of Coq, especially the way linear operators and computer programs are specified in it.
	In NQPV, operators are given concretely as \texttt{numpy} matrices, and programs are represented in a natural way familiar to ordinary Python programmers. As mentioned at the beginning of this section, we take this easy path to implement the prototype since our main purpose is to illustrate the utility of our logic system, rather than a full-blown verification tool. 
        \subsubsection*{Performance}
        The performance of CoqQ is determined by the proofs provided, while that of NQPV is determined by the calculation backend, which in the worst case is exponential in the number of qubits.
        In comparison, it takes a few seconds to verify the general Grover algorithm in CoqQ, and 90 seconds for the 13-qubit Grover algorithm in NQPV.

\balance
			
\section{Conclusion}	\label{sec:conclusion}
	In this paper, we consider quantum programs where nondeterministic choices are allowed. Denotational semantics of such programs is given by lifting semantics of deterministic quantum programs, and we argue that this is the only natural way to define nondeterminism in the quantum setting. 
	For the purpose of verification, we take sets of hermitian operators on the associated Hilbert space to be the assertions of quantum states,
	and propose two proof systems, for partial and total correctness respectively. We show that they are both sound and relatively complete with respect to the corresponding correctness notions. Simple quantum algorithms and protocols, such as the three-qubit bit-flip quantum error correction scheme, Deutsch algorithm, and a nondeterministic quantum walk, are analysed to demonstrate the utility of our proof systems. A lightweight prototype of a proof assistant is implemented to aid in the automated reasoning of nondeterministic quantum programs.
	
	For future works, we are going to investigate how to make use of the nondeterministic choice construct and the verification technique proposed in this paper for quantum program refinement. How to incorporate angelic nondeterminism into the picture is also an interesting topic for future study.  

%\clearpage
\begin{acks}
The authors thank Profs Mingsheng Ying and Sanjiang Li for inspiring discussions. This work is partially supported by 
\grantsponsor{2018YFA0306704}{the National Key R\&D Program of China}{} under Grant No
\grantnum{2018YFA0306704}{2018YFA0306704} and
\grantsponsor{DP180100691}{the Australian Research Council}{https://www.arc.gov.au/}
under Grant No.
\grantnum{DP220102059}{DP220102059}. YX was also partially funded by \grantsponsor{}{the Sydney Quantum Academy}{https://sydneyquantum.org/}.

%\grantsponsor{⟨sponsorID⟩}{⟨name⟩}{⟨url⟩}
%\grantnum[⟨url⟩]{⟨sponsorID⟩}{⟨number⟩}

\end{acks}

%\section*{Appendices}
\appendix

 \section{Weakest (liberal) precondition semantics}\label{sec:wps}

	This section presents an alternative semantics for nondeterministic quantum programs in terms of the weakest (liberal) preconditions. It turns out to be equivalent to the denotational semantics given in Sec.~\ref{sec:denotation}. More importantly, the weakest (liberal) precondition semantics provides a powerful proof technique for the soundness and completeness of our logic systems, as shown in the next section. 
			
 	Let $S\in \prog$. The \emph{weakest precondition semantics} $wp.S$ and \emph{weakest liberal precondition semantics} $wlp.S$ of $S$ are both mappings in
 		$2^\a\ra  2^\a$
 		defined inductively in Fig.~\ref{tbl:wpsemantics}. To simplify notation, we use $xp$ to denote $wp$ or $wlp$ whenever it is applicable for both of them.

	The definitions are similar to those of deterministic quantum programs presented in~\cite{ying2012floyd}. Again, because of the possible nondeterministic choices in the loop body $S$, we have incorporated schedulers $\eta\in \sem{S}^{\N}$ in the weakest (liberal) precondition semantics of $\pwstm$. 
	
	The following lemma shows a duality between the denotational and weakest (liberal) precondition semantics of quantum programs.

 	{\renewcommand{\arraystretch}{2}
	 		\begin{figure*}[t]
			 				\centering
			 				\begin{tabular}{l}
				 					\begin{tabular}{ll}
					 						$xp.\sskip.M = M$ 
					 						&$xp.(\bar{q}\apply U).M  = U_{\bar{q}}^\dag M U_{\bar{q}}$\\
					 						$xp.(\bar{q}:=0).M = \displaystyle\sum_{i=0}^{2^n-1} \quiz M\quzi$ &$xp.S.\qassert = \displaystyle\bigcup_{M\in \qassert} xp.S.M$\\
					 						$xp.(S_0; S_1).M = xp.S_0.(xp.S_1.M)$\hspace{6em} 
					 						&$xp.(S_0\ \square\ S_1).M = xp.S_0.M \cup xp.S_1.M$\\
					 						\smallskip
					 						$wlp.\abort.M = \{I\}$ & $wp.\abort.M = \{\z\}$\\
					 					\end{tabular}\\ 
				 					\begin{tabular}{l}
					 												\smallskip
					 						$xp.(\pmstm).M = \p^1_{\bar{q}}\left(xp.S_1.M\right) + \p^0_{\bar{q}} \left(xp.S_0.M\right)$\\	
					 																		\smallskip
					 						$wlp.(\pwstm).M = \left\{\displaystyle\bigwedge_{i\geq 0} M_i^\eta : \eta \in \sem{S}^\N\right\}$ where for each $\eta$,
					 						$M_0^\eta \define I$, and for any $i\geq 0$,\\
					 						\hspace{10em}
					 						$M_{i+1}^\eta  = \p^0_{\bar{q}}(M) +  \p^1_{\bar{q}} \left(\eta_1^\dag \left(M_i^{\eta^\ra}\right) + I - \eta_1^\dag(I)\right)$
					 						\\ 
					 												\smallskip
					 						$wp.(\pwstm).M =\left\{\displaystyle\bigvee_{i\geq 0} M_i^\eta : \eta \in \sem{S}^\N\right\}$ where for each $\eta$,
					 						$M_0^\eta \define \z$, and for any $i\geq 0$,\\
					 						\hspace{10em}
					 						$M_{i+1}^\eta  = \p^0_{\bar{q}}(M) +  \p^1_{\bar{q}} \left(\eta_1^\dag \left(M_i^{\eta^\ra}\right)\right)$			
					 					\end{tabular}			
				 				\end{tabular}
		 			\caption{Weakest (liberal) precondition semantics for quantum programs, where $xp\in \{wp, wlp\}$. 
			 			}
		 			\label{tbl:wpsemantics}
		 		\end{figure*}
	 	}

 	\begin{lemma}\label{lem:wpwlp}
	 		Let $S\in \prog$, $\qstate\in \d(\h_{\QVar})$, $M$ be a quantum predicate, and $\qassert$ a quantum assertion. Then
	 		\begin{enumerate}
		 			\item $wp.S.M = \left\{\e^\dag (M) : \e\in \sem{S}\right\}$;
		 			\item $wlp.S.M = \left\{\e^\dag (M) + I - \e^\dag (I) : \e\in \sem{S}\right\}$;
		 			\item $	\Exp(\qstate \models wp.S.\qassert)= \inf \left\{	\Exp(\sigma \models \qassert) : \sigma\in \sem{S}(\qstate) \right\}$; 
		 			\item $	\Exp(\qstate \models wlp.S.\qassert)=$\\ $\inf \left\{	\Exp(\sigma \models \qassert) + \tr(\rho) - \tr(\sigma) : \sigma\in \sem{S}(\qstate) \right\}$.
		 		\end{enumerate}
	 	\end{lemma}
 	\begin{proof}
	 		We prove clause (1)  by induction on the structure of $S$. The basis cases are easy from the definition. For the induction step, we only show the following two cases as examples.
	
	 		\begin{itemize}
		 			\item Let $S\define \pmstm$. Then 
		 			\begin{align*}
			 				&wp.S.M\\
			 				& = \p^1_{\bar{q}}(wp.S_1.M) + \p^0_{\bar{q}} (wp.S_0.M)\\
			 				&=  \left\{\p^1_{\bar{q}}\circ \f^\dag (M) : \f\in \sem{S_1}\right\} + \left\{\p^0_{\bar{q}}\circ \e^\dag (M) : \e\in \sem{S_0}\right\}\\
			 				&=  \left\{(\f\circ \p^1_{\bar{q}}+\e \circ \p^0_{\bar{q}})^\dag (M) : \e\in \sem{S_0}, \f\in \sem{S_1}\right\}\\
			 				&=  \left\{\g^\dag (M) : \g\in \sem{S}\right\}.
			 			\end{align*}
                        \begin{sloppypar}
			 			\item Let $S\define \whilestm{S'}$. For any $\eta\in \sem{S'}^\N$ and $i\geq 0$, let $M_i^\eta$ be defined as in Fig.~\ref{tbl:wpsemantics} for $wp.(\whilestm{S'}).M$, and $\f_i^\eta$ be defined as in Figure~\ref{fig:densemantics} for $\sem{\whilestm{S'}}$.	First, from Lemma~\ref{lem:whilesem} it is easy to show by induction on $i$ that
                        \end{sloppypar}
			 			\begin{equation}\label{eq:mfeq}
				 			\forall i\geq 0, \ \forall \eta\in \sem{S'}^\N, \ 
				 			M_i^\eta = (\f_i^\eta)^\dag (M).
				 		\end{equation}
			 			Thus we have
			 			\[
			 			\bigvee_{i\geq 0} M_i^\eta = \bigvee_{i\geq 0} \left(\f_i^\eta\right)^\dag (M) = \left(\bigvee_{i\geq 0} \f_i^\eta\right) ^\dag (M),
			 			\]
			 			and so
			 \[
			 			wp.S.M = \left\{\bigvee_{i\geq 0} M_i^\eta : \eta\in \sem{S'}^\N\right\} =  \left\{\e^\dag(M) : \e\in \sem{S}\right\}
			 \]	
			 from the definition of $\sem{S}$.
			 	\end{itemize}
		
		 		The proof for clause (2) is similar; the only difference in proving the case for $\whilestm{S'}$ is that instead of Eq.~\eqref{eq:mfeq} we have to prove 
		 					\[
		 		\forall i\geq 0, \ \forall \eta\in \sem{S'}^\N, \ 
		 		I - M_i^\eta = (\f_i^\eta)^\dag (I-M)
		 		\]
                    \begin{sloppypar}\noindent
		 		where $M_i^\eta$'s are defined in a similar way as in Fig.~\ref{tbl:wpsemantics} but for $wlp.(\whilestm{S'}).M$. However, this is also easy by induction on $i$.
		          \end{sloppypar}
		 		For clause (3), we calculate
		 		\begin{align*}
			 			\Exp(\qstate \models wp.S.\qassert) &= \inf \left\{	\tr(M\qstate) : N\in \qassert, M\in wp.S.N \right\}\\
			 			&= \inf \left\{	\tr(\e^\dag(N)\qstate) : N\in \qassert, \e\in \sem{S} \right\}\\
			 			&= \inf \left\{	\tr(N\e(\qstate)) : N\in \qassert, \e\in \sem{S} \right\}\\
			 			&= \inf \left\{	\tr(N\sigma) :  N\in \qassert, \sigma\in \sem{S}(\qstate) \right\}\\
			 			&= \inf \left\{	\Exp(\sigma \models \qassert) : \sigma\in \sem{S}(\qstate) \right\}
			 		\end{align*}
		 		where the second equality is from clause (1). Finally, clause (4) follows from (2) with a similar argument.
		 	\end{proof}
	
	 	The next lemma shows that the weakest (liberal) precondition of a while program is a fixed point of some functor on $\a$.
	 	
	 	\begin{lemma}\label{lem:recpredicate} \begin{sloppypar} For any quantum predicate $M$, let $\qassert = xp.(\pwstm).M$. Then\end{sloppypar}
		 		\[
		 		\qassert = \p^0_{\bar{q}}(M) + \p^1_{\bar{q}}(xp.S.\qassert).
		 		\]
		 	\end{lemma}
	 	\begin{proof} 
		 	 Let $\while \define \pwstm$. Then
		 		\begin{align*}
			 			&wp.\while.M \\
			 			&= \left\{\e^\dag(M) : \e\in \sem{\while}\right\} \\
			 			&= \left\{\p^0_{\bar{q}}(M) + \p^1_{\bar{q}}\circ
			 			\f^\dag \circ \g^\dag(M)  : \f\in \sem{S}, \g\in \sem{\while}\right\}\\
			 			&= \p^0_{\bar{q}}(M) + \p^1_{\bar{q}}\left(wp.S.\left\{
			 			\g^\dag(M)  :  \g\in \sem{\while}\right\}\right)\\
			 			&= \p^0_{\bar{q}}(M) + \p^1_{\bar{q}}\left(wp.S.\left(wp.\while.M\right)\right)
			 		\end{align*}
		 		where the second equality is from Eq.~\eqref{eq:whilesem} while the other ones follow from Lemma~\ref{lem:wpwlp}(1).
		 		The case for $wlp$ is similar.
		 	\end{proof}

	To conclude this section, we prove a lemma which is useful in our later discussion.
	
	\begin{lemma}\label{lem:rulecond} For any correctness formula $\ass{\qassert}{S}{\qassertp}$,
		\begin{enumerate}
			\item $\models_{\tot} \ass{\qassert}{S}{\qassertp}$ if and only if $\qassert \leinf wp.S.\qassertp$; 
			\item $\models_{\pal} \ass{\qassert}{S}{\qassertp}$ if and only if $\qassert \leinf wlp.S.\qassertp$.
		\end{enumerate}
		\begin{proof}
			For clause (1), we compute 
			\begin{align*}
				&\models_{\tot} \ass{\qassert}{S}{\qassertp}\\
				 &\ \Leftrightarrow \ \forall \rho, \Exp(\qstate \models \qassert)\leq \inf \left\{	\Exp(\sigma \models \qassertp) : \sigma\in \sem{S}(\qstate) \right\} \\
				&\ \Leftrightarrow \ \forall \rho, \Exp(\qstate \models \qassert)\leq \Exp(\qstate \models wlp.S.\qassertp) \\
				&\ \Leftrightarrow \ \forall \rho, \inf \left\{\tr(M\qstate) : M\in \qassert\right\}\leq \inf\left\{\tr(N\qstate) : N\in wlp.S.\qassertp\right\} \\
				&\ \Leftrightarrow \ \qassert \leinf wp.S.\qassertp
			\end{align*}
		where the second equivalence is from Lemma~\ref{lem:wpwlp}(3) while the third one from Eq~\eqref{eq:defexp}.
			The case for $wlp$ in clause (2) is similar.
		\end{proof}
	\end{lemma}

\section{Proof details}\label{app:proof}
\subsection{Proof of Theorem~\ref{thm:psc}}
	 		Soundness:  We need only to show that each rule in Fig.~\ref{tbl:psystem} is valid in the sense of partial correctness. Take the rule (While) as an example; the others are simpler. From (Union), it suffices to consider the special case when $\qassertp = \{M\}$ for some quantum predicate $M$. Let
	 		$$\models_{\pal}  \ass{\qassert}{S}{\p^0_{\bar{q}}(M) + \p^1_{\bar{q}}(\qassert)}.$$ Then $\qassert\leinf wlp.S.\left(\p^0_{\bar{q}}(M) + \p^1_{\bar{q}}(\qassert)\right)$. We now prove by induction on $i$ that
	 		\[
	 		\forall i\geq 0, \forall \eta\in \sem{S}^{\N},\ \p^0_{\bar{q}}(M) + \p^1_{\bar{q}}(\qassert)  \leinf M^\eta_i
	 		\]
	 		where $M^\eta_i$ is defined as in Fig.~\ref{tbl:wpsemantics} for the $wlp$ semantics of $\while\define\pwstm$. The case when $i=0$ is trivial. Then we calculate
	 		\begin{align*}
		 			&\p^0_{\bar{q}}(M) + \p^1_{\bar{q}}(\qassert)\\ 	& \leinf  \p^0_{\bar{q}}(M) + \p^1_{\bar{q}}\left(wlp.S.\left(\p^0_{\bar{q}}(M) + \p^1_{\bar{q}}(\qassert)\right)\right)\\
		 			& \leinf  \p^0_{\bar{q}}(M) + \p^1_{\bar{q}}\left(wlp.S.M^{\eta^{\ra}}_i\right)\\
		 			& \leinf  \p^0_{\bar{q}}(M) +  \p^1_{\bar{q}} \left(\eta_1^\dag \left(M_i^{\eta^{\ra}}\right) + I - \eta_1^\dag(I)\right)\\
		 			& = M^\eta_{i+1},
		 		\end{align*}
	 		where the first inequality follows from Lemma~\ref{lem:propleinf}(1), the second one from the induction hypothesis and that $wlp.S$ is monotonic with respect to $\leinf$, and the third one  from Lemma~\ref{lem:wpwlp} and that $\eta_1\in \sem{S}$.
	 		Thus $$\p^0_{\bar{q}}(M) + \p^1_{\bar{q}}(\qassert) \leinf wlp.\while.M,$$ and so 
	 		$$\models_{\pal} \ass{\p^0_{\bar{q}}(M) + \p^1_{\bar{q}}(\qassert)}{\while}{M}$$ from Lemma~\ref{lem:rulecond}.
	
	 		Completeness: By Lemma~\ref{lem:rulecond} and the (Imp) rule, it suffices to show that for any $\qassert$ and $S'$,
	 		$$\vdash_{\pal} \ass{wlp.S'.\qassert}{S'}{\qassert}.$$
	 		Again, we take the case for loops as an example.  For any $M\in \qassert$, let $\qassertp \define wlp.\while.M$.
	 		By induction, we have 
	 		$\vdash_{\pal} \ass{wlp.S.\qassertp}{S}{\qassertp}.$
	 		Note from Lemma~\ref{lem:recpredicate} that
	 		$	\qassertp = \p^0_{\bar{q}}(M) + \p^1_{\bar{q}}\left(wlp.S.\qassertp\right).$
	 		Then we have
	 		$\vdash_{\pal} \ass{\qassertp}{\while}{M}$ by rule (While).
	 		Finally, the desired result  
	 		$$\vdash_{\pal} \ass{wlp.\while.\qassert}{\while}{\qassert}$$
	 		follows from (Union) and Lemma~\ref{lem:propleinf}(2).

\subsection{Proof of Theorem~\ref{thm:total}}
                    \begin{sloppypar}
	 		Soundness:  We need only to show that each rule of the proof system is valid in the sense of total correctness. Take rule (WhileT) as an example. Again, from (Union), it suffices to consider the special case when $\qassertp = \{M\}$ for some quantum predicate $M$. Let 
	 		\[
	 		\models_{\tot}	\ass{\qassert}{S}{\p^0_{\bar{q}}(M) + \p^1_{\bar{q}}(\qassert)},
	 		\]
	 		and $\{R_i^\eta : i\geq 0, \eta\in \sem{S}^{\N}\}$  be a $\p^0_{\bar{q}}(M) + \p^1_{\bar{q}}(\qassert)$-ranking assertion for $\pwstm$.  
	 		Then $\qassert\leinf wp.S.\left(\p^0_{\bar{q}}(M) + \p^1_{\bar{q}}(\qassert)\right)$.
	 		We now prove by induction on $i$ that
	 		$$ \forall i\geq 0, \forall \eta\in \sem{S}^{\N},\ \p^0_{\bar{q}}(M) + \p^1_{\bar{q}}(\qassert)  \leinf M^\eta_i + R_i^\eta$$
	 		where $M^\eta_i$ is defined as in Fig.~\ref{tbl:wpsemantics} for the $wp$ semantics of $\pwstm$. The case when $i=0$ is from the assumption that $\p^0_{\bar{q}}(M) + \p^1_{\bar{q}}(\qassert)\leinf R_0^\eta$. Then we calculate
                    \end{sloppypar}
	 		\begin{align*}
		 			&\p^0_{\bar{q}}(M) + \p^1_{\bar{q}}(\qassert)\\ 	& \leinf  \p^0_{\bar{q}}(M) + \p^1_{\bar{q}}\left(wp.S.\left(\p^0_{\bar{q}}(M) + \p^1_{\bar{q}}(\qassert)\right)\right)\\
		 			& \leinf  \p^0_{\bar{q}}(M) + \p^1_{\bar{q}}\left(wp.S.\left(M^{\eta^{\ra}}_i + R^{\eta^{\ra}}_i\right)\right)\\
		 			& \leinf  \p^0_{\bar{q}}(M) +  \p^1_{\bar{q}} \left(\eta_1^\dag \left(M_i^{\eta^{\ra}}\right)\right) + \p^1_{\bar{q}} \left(\eta_1^\dag \left(R_i^{\eta^{\ra}}\right)\right) \\
		 			& \leinf M^\eta_{i+1} + R^\eta_{i+1},
		 		\end{align*}
	 		where  the first inequality follows from Lemma~\ref{lem:propleinf}(1), the second one from the induction hypothesis and that $wp.S$ is monotonic with respect to $\leinf$, the third one from Lemma~\ref{lem:wpwlp} and that $\eta_1\in \sem{S}$, and the last one from Eq.~\eqref{eq:compact}.
	 		Thus $$\p^0_{\bar{q}}(M) + \p^1_{\bar{q}}(\qassert) \leinf wp.(\pwstm).M$$
	 		by noting that  $\bigwedge_i R_i^\eta = \z$ for any $\eta$, and so 
	 		$$\models_{\tot} \ass{\p^0_{\bar{q}}(M) + \p^1_{\bar{q}}(\qassert)}{\pwstm}{M}$$ from Lemma~\ref{lem:rulecond}.
	
	 		Completeness: By Lemma~\ref{lem:rulecond} and the (Imp) rule, it suffices to show that for any $\qassert$ and $S'$,
	 		$$\vdash_{\tot} \ass{wp.S'.\qassert}{S'}{\qassert}.$$
	 		Again, we take the case for while loops as an example. Let $\while \define \pwstm$ and $M\in \qassert$.
	 		By induction, we have 
	 		$\vdash_{\tot} \ass{wp.S.\qassertp}{S}{\qassertp}$ where  
	 		$\qassertp \define wp.\while.M$.
	 		Note also from Lemma~\ref{lem:recpredicate} that
	 		$\qassertp = \p^0_{\bar{q}}(M) + \p^1_{\bar{q}}\left(wp.S.\qassertp\right).$ To finish the proof, we have to construct a $\qassertp$-ranking assertion for $\while$.
	
	 		For any $\eta\in \sem{S}^{\N}$ and $k\geq 0$, let 
	 				\begin{equation}\label{eq:induc}
		 			R_k^\eta=
		 			\sum_{i=k}^\infty \p^1_{\bar{q}}\circ \eta_1^\dag\circ \ldots \circ \p^1_{\bar{q}}\circ \eta_i^\dag \circ \p^0_{\bar{q}}(I).
		 		\end{equation}
	 	We would like to show that the set of $R_k^\eta$ satisfy the conditions in Definition~\ref{def:rank}. First, we prove by induction on $i$ that 
	 	\begin{equation}\label{eq:indrank}
		 		\forall i\geq 0, \forall \eta\in \sem{S}^{\N},\ M_i^\eta \le R_0^\eta
		 	\end{equation}
	 		where $M^\eta_i$ is defined as in Fig.~\ref{tbl:wpsemantics} for $wp.\while.M$,
	 	which then implies that $\qassertp\leinf R_0^\eta$. The base case of Eq.~\eqref{eq:indrank} where $i=0$ is trivial since $M_0^\eta = \z$. 		We further calculate that
	 	\begin{align*}
		 		M_{i+1}^\eta  &= \p^0_{\bar{q}}(M) +  \p^1_{\bar{q}} \left(\eta_1^\dag \left(M_i^{\eta^\ra}\right)\right)\\
		 		&\le \p^0_{\bar{q}}(I) +  \p^1_{\bar{q}} \left(\eta_1^\dag \left(R_0^{\eta^\ra}\right)\right)\\
		 		& = R_0^\eta.
		 	\end{align*} 
	 	For the second condition of Definition~\ref{def:rank}, we note that  from Eq.~\eqref{eq:induc}, $R_{k+1}^\eta \le R_k^\eta$ and $\bigwedge_k R_k^\eta = \z$. Furthermore, it is easy to see that $
	 	R_{k+1}^{\eta} = \p^1_{\bar{q}} \left(\eta_1^\dag \left(R_k^{\eta^{\ra}}\right)\right)$.

	 	Now using rule (WhileT), we have		
	 	$\vdash_{\tot} \ass{\qassertp}{\while}{M}.$
	 		Finally, the desired result  
	 		$$\vdash_{\tot} \ass{wp.\while.\qassert}{\while}{\qassert}$$
	 		follows from (Union) and Lemma~\ref{lem:propleinf}(2).

%\clearpage
%%%%%%%%%%%%%%%%%%%%%%%%%%%%%%%%%%%%%%%%%%
%   The artifact appendix
%%%%%%%%%%%%%%%%%%%%%%%%%%%%%%%%%%%%%%%%%%
\input{artifact_appendix.tex}

%\clearpage
\bibliographystyle{ACM-Reference-Format}
\bibliography{ref}

\end{document}

%% file: pmymacro.tex
\newcommand {\tot} {\mathit{tot}}
\newcommand {\pal} {\mathit{par}}

\def\>{\ensuremath{\rangle}}
\def\<{\ensuremath{\langle}}

\def\-{\ensuremath{\textrm{-}}}

\def\apply{\mathrel{*\!\!=}}

\def\qVar{\ensuremath{\mathit{qv}}}

\def\QVar{\ensuremath{\mathcal{V}}}

\def\Exp{\mathit{Exp}}

\def\h{\ensuremath{\mathcal{H}}}
\def\p{\ensuremath{\mathcal{P}}}
\def\l{\ensuremath{\mathcal{L}}}
\def\g{\ensuremath{\mathcal{G}}}
\def\lh{\ensuremath{\mathcal{L(H)}}}
\def\dh{\ensuremath{\mathcal{D(H})}}
\def\dhv{\ensuremath{\d(\h_{\QVar})}}

\def\m{\ensuremath{\mathcal{M}}}
\def\u{\ensuremath{\mathcal{U}}}

\def\s{\ensuremath{\mathcal{S}}}

\def\u{\ensuremath{\mathcal{U}}}

\def\x{\ensuremath{\mathcal{X}}}

\def\z{0}

\def\ra{\ensuremath{\rightarrow}}
\def\a{\ensuremath{\mathcal{A}}}

\def\e{\ensuremath{\mathcal{E}}}
\def\f{\ensuremath{\mathcal{F}}}
\def\l{\ensuremath{\mathcal{L}}}

\def\N{\mathbb{N}}

\def\quiz{\ensuremath{|i\>_{\bar{q}}\<0|}}
\def\quzi{\ensuremath{|0\>_{\bar{q}}\<i|}}

\def\d{\ensuremath{\mathcal{D}}}
\def\dh{\ensuremath{\mathcal{D(H)}}}
\def\lh{\ensuremath{\mathcal{L(H)}}}
\def\le{\ensuremath{\sqsubseteq}}
\def\ge{\ensuremath{\sqsupseteq}}

\def\bf{\mathbf{f}}

\def\next{\mathcal{X}}

\renewcommand{\theenumi}{(\arabic{enumi})}

\newcommand{\tr}{{\rm tr}}

\newcommand{\ass}[3]{\left\{#1\right\}\ #2\ \left\{#3\right\}}

\newcommand {\true} {\ensuremath{{\mathbf{true}}}}
\newcommand {\false} {\ensuremath{{\mathbf{false}}}}
\newcommand {\abort}{\ensuremath{{\mathbf{abort}}}}
\newcommand {\sskip} {\mathbf{skip}}

\newcommand {\then} {\ensuremath{\mathbf{then}}}
\newcommand {\eelse} {\ensuremath{\mathbf{else}}}
\newcommand {\while} {\ensuremath{\mathbf{while}}}
\newcommand {\ddo} {\ensuremath{\mathbf{do}}}
\newcommand {\pend} {\ensuremath{\mathbf{end}}}

\newcommand {\mymeas} {\mathbf{measure}}

\newcommand {\iif} {\mathbf{if}}

\def\pmstm{\iif\ \m[\bar{q}]\ \then\ S_1\ \eelse\ S_0\ \pend}
\def\pwstm{\while\ \m[\bar{q}]\ \ddo\ S\ \pend}

\newcommand\whilestm[1]{\while\ \m[\bar{q}]\ \ddo\ #1\ \pend}

\newcommand{\define}{\ensuremath{\triangleq}}

\newcommand{\id}{\mathcal{I}}
\def\qstate{\rho}
\def\qassert{\Theta}
\def\qassertp{\Psi}

\def\Exp{\mathrm{Exp}}

\def\supoprset{\mathbb{E}}

\def\leinf{\le_{\mathit{inf}}}

\def\qstatesh#1{\d(\mathcal{H}_{#1})}

\newcommand\prog{\mathit{Prog}}

\def\l{\mathcal{L}}

\def \Rm#1{\mbox{\rm #1}}
\def \lsem      {\raise1pt\hbox{\Rm {[\kern-.12em[}}}
\def \rsem      {\raise1pt\hbox{\Rm {]\kern-.12em]}}}
\def \sem#1{\mbox{\lsem$#1$\rsem}}

%% file: artifact_appendix.tex
% LaTeX template for Artifact Evaluation V20201122
%
% Prepared by Grigori Fursin with contributions from Bruce Childers,
%   Michael Heroux, Michela Taufer and other colleagues.
%
% See examples of this Artifact Appendix in
%  * SC'17 paper: https://dl.acm.org/citation.cfm?id=3126948
%  * CGO'17 paper: https://www.cl.cam.ac.uk/~sa614/papers/Software-Prefetching-CGO2017.pdf
%  * ACM ReQuEST-ASPLOS'18 paper: https://dl.acm.org/citation.cfm?doid=3229762.3229763
%
% (C)opyright 2014-2022
%
% CC BY 4.0 license
%

% \documentclass{sigplanconf}

% \usepackage{hyperref}

% \begin{document}

% \special{papersize=8.5in,11in}

%%%%%%%%%%%%%%%%%%%%%%%%%%%%%%%%%%%%%%%%%%%%%%%%%%%%
% When adding this appendix to your paper, 
% please remove above part
%%%%%%%%%%%%%%%%%%%%%%%%%%%%%%%%%%%%%%%%%%%%%%%%%%%%

% \appendix
\section{Artifact Appendix}
\renewcommand {\bf} {\bfseries}
%%%%%%%%%%%%%%%%%%%%%%%%%%%%%%%%%%%%%%%%%%%%%%%%%%%%%%%%%%%%%%%%%%%%%
\subsection{Abstract}

NQPV is a verification assistant prototype of nondeterministic quantum programs. It implements the verification logic of partial correctness in the numerical form, with soundness guaranteed by the theory and experiments.
It is a Python project, depending on \texttt{numpy}, \texttt{cvxpy} and \texttt{ply} packages.
The three examples in our paper: the three-qubit bit-flip quantum error correction scheme, Deutsch algorithm and quantum walk, are encoded and verified numerically.
The evaluation is integrated into a step-by-step Jupyter notebook. It shows that the tool can check the program syntax, automatically calculate the verification condition, check whether it is satisfied by the precondition, and return a proof outline of the correctness formula if so.
Compared to other verification tools based on theorem provers like Coq or Isabelle, NQPV is weaker in expressiveness but stronger in automation. 
The whole evaluation can be conducted using a computer with 32 GB memory.

\subsection{Artifact check-list (meta-information)}

{\em 
% Obligatory. Use just a few informal keywords in all fields applicable to your artifacts
% and remove the rest. This information is needed to find appropriate reviewers and gradually 
% unify artifact meta information in Digital Libraries.
}

{\small
\begin{itemize}
  \item {\bf Algorithm: } Verification of nondeterministic quantum  programs.
  %\item {\bf Program: }
  %\item {\bf Compilation: }
  %\item {\bf Transformations: }
  %\item {\bf Binary: }
  %\item {\bf Model: }
  %\item {\bf Data set: }
  \item {\bf Run-time environment: } Python 3. The project is developed with Python 3.9.
  %\item {\bf Hardware: }
  %\item {\bf Run-time state: not sensitive}
  %\item {\bf Execution: } 10 minutes 
  \item {\bf Metrics: } Expressiveness of the programs and quantum assertions; Automation of the verification; Time and memory consumption for the verification of larger programs.
  \item {\bf Output: } Numerical verification conditions saved in a computer file.
  \item {\bf Experiments: } Integrated in a Jupyter notebook.
  \item {\bf How much disk space required (approximately)?: } Only the performance test is resource consuming. 8 GB is sufficient for the experiments here.
  \item {\bf How much time is needed to prepare workflow (approximately)?: } 5 minutes.
  \item {\bf How much time is needed to complete experiments (approximately)?: } 15 minutes to go through the Jupyter notebook.
  \item {\bf Publicly available: } Yes, on GitHub, Zenodo and PyPI.
  \item {\bf Code licenses (if publicly available)?: } Apache-2.0
  %\item {\bf Data licenses (if publicly available)?: }
  %\item {\bf Workflow framework used?: }
  \item {\bf Archived (provide DOI)?: } https://doi.org/10.5281/zenodo.7564087    
\end{itemize}
}

%%%%%%%%%%%%%%%%%%%%%%%%%%%%%%%%%%%%%%%%%%%%%%%%%%%%%%%%%%%%%%%%%%%%%
\subsection{Description}

\subsubsection{How to access}

This is a Python project on GitHub: \url{https://github.com/LucianoXu/NQPV}.
For evaluation purpose, please clone the "Article Release".
The project itself is about 1 MB in size.

This project is also uploaded on PyPI as a package: \url{https://pypi.org/project/NQPV/}.

\subsubsection{Hardware dependencies}
Only the performance test is resource consuming.
A laptop computer with 32 GB memory and 8 GB disk storage is sufficient for the evaluation.

\subsubsection{Software dependencies}
Python packages: \texttt{numpy} for matrix calculation, \texttt{cvxpy} for the SDP solver, and \texttt{ply} for syntax parsing.

%\subsubsection{Data sets}

%\subsubsection{Models}

%%%%%%%%%%%%%%%%%%%%%%%%%%%%%%%%%%%%%%%%%%%%%%%%%%%%%%%%%%%%%%%%%%%%%
\subsection{Installation}
Experiments with the source code:
\begin{enumerate}
\item Prepare a Python 3 environment (This project was developed with Python 3.9).
\item Install package dependencies with this command: 

        \texttt{pip install cvxpy ply}
        
\item Clone the NQPV project or download the "Article Release".
\item Open a console at the root folder of the project and run the test script with

        \texttt{python test\_install.py}

        If no error is reported, then the installation is successfully completed.
\end{enumerate}
NQPV as a Python package: run the command \texttt{pip install NQPV}
%%%%%%%%%%%%%%%%%%%%%%%%%%%%%%%%%%%%%%%%%%%%%%%%%%%%%%%%%%%%%%%%%%%%%
\subsection{Experiment workflow}
Find the Jupyter notebook \texttt{evaluate.ipynb} in the root folder of source. 
This notebook contains a brief introduction to the tool and integrates the verification experiments of all examples in our paper.

Every individual experiment is done with a \texttt{.nqpv} file containing the NQPV program code, 
and a \texttt{.py} python script to prepare the operators and invoke the verification.

%%%%%%%%%%%%%%%%%%%%%%%%%%%%%%%%%%%%%%%%%%%%%%%%%%%%%%%%%%%%%%%%%%%%%
\subsection{Evaluation and expected results}

The three examples in the paper: the three-qubit bit-flip quantum error correction scheme, Deutsch algorithm and quantum walk are encoded as program codes in the Jupyter notebook. The notebook also contains some examples to demonstrate the ability of NQPV to reject false correctness claims.
The required operators (special unitary operators or measurements) are prepared and saved as \texttt{numpy} matrices.

It is expected that the verification is automatically completed and a proof outline is printed. We can check every intermediate condition in the proof outline using the \texttt{show} command. The verification condition of each example is compared with the given precondition.

Finally, a performance test of NQPV on larger qubit numbers is presented. For the Grover algorithm, NQPV takes several minutes and 32 GB memory to finish the 
verification of the 13 qubit case.

%%%%%%%%%%%%%%%%%%%%%%%%%%%%%%%%%%%%%%%%%%%%%%%%%%%%%%%%%%%%%%%%%%%%%
\subsection{Experiment customisation}
The examples in the notebook demonstrate the syntax of describing a verification task in NQPV.

To conduct a customised experiment, we can:
\begin{enumerate}
    \item Encode the verification task in the \texttt{example.nqpv} file. 
        It contains the program to be verified, the pre- and post-conditions, and the loop invariants.
    \item Prepare a python script to produce the operators used in \texttt{example.nqpv} as binary \texttt{numpy} matrix files.
        Then invoke the verification using the commend

        \texttt{nqpv.verify("example.nqpv")}
    \item Run the python script to conduct the verification.
        
    \item After obtaining the proof outline, modify \texttt{example.nqpv} and repeat to print the conditions during verification.
\end{enumerate}

 To avoid path errors, the above experiments should be conducted in the root folder of NQPV source, or use NQPV installed from PyPI.

% %%%%%%%%%%%%%%%%%%%%%%%%%%%%%%%%%%%%%%%%%%%%%%%%%%%%%%%%%%%%%%%%%%%%%
\subsection{Notes}
MacOS users may encounter an installation error due to the package \texttt{cmake}. In this case, try to run \texttt{pip install cmake} first, then add \texttt{cmake} to PATH manually, and finally install NQPV. 

We recommend using \texttt{conda} to create a new experimental environment. The deployment was also tested on a Linux cloud server.

% %%%%%%%%%%%%%%%%%%%%%%%%%%%%%%%%%%%%%%%%%%%%%%%%%%%%%%%%%%%%%%%%%%%%%
% \subsection{Methodology}

% Submission, reviewing and badging methodology:

% \begin{itemize}
%   \item \url{https://www.acm.org/publications/policies/artifact-review-badging}
%   \item \url{http://cTuning.org/ae/submission-20201122.html}
%   \item \url{http://cTuning.org/ae/reviewing-20201122.html}
% \end{itemize}

%%%%%%%%%%%%%%%%%%%%%%%%%%%%%%%%%%%%%%%%%%%%%%%%%%%%
% When adding this appendix to your paper, 
% please remove below part
%%%%%%%%%%%%%%%%%%%%%%%%%%%%%%%%%%%%%%%%%%%%%%%%%%%%

% \end{document}